%% file: SM_LASSO.tex
\documentclass[conference]{IEEEtran}
\ifCLASSINFOpdf
\else
\fi
\usepackage{amsmath,dsfont,bbm,epsfig,amssymb,amsfonts,amstext}\usepackage{verbatim,amsopn,cite,subfigure,multirow,multicol,lipsum,xfrac}
\usepackage{amsthm}
\usepackage{mathtools,amsthm}
\usepackage{balance}
\usepackage{url}
\usepackage{amsfonts}
\usepackage{epsfig}
\usepackage{caption}
\usepackage{psfrag}	
\usepackage{etoolbox}
\usepackage[Algorithm,ruled]{algorithm}
\usepackage{algorithmicx}
\usepackage{algpseudocode}
\usepackage{pifont}
\usepackage[utf8]{inputenc}
\usepackage[T1]{fontenc}  
\usepackage[nolist]{acronym}
\usepackage{paralist}
\usepackage{enumitem}
\usepackage{bbm}
\usepackage[process=auto]{pstool}
\usepackage{tikz,pgfplots}
\usetikzlibrary{shapes,arrows}
\include{commands}
\begin{document}

\begin{acronym}
\acro{mimo}[MIMO]{multiple-input multiple-output}
\acro{csi}[CSI]{channel state information}
\acro{awgn}[AWGN]{additive white Gaussian noise}
\acro{iid}[i.i.d.]{independent and identically distributed}
\acro{ut}[UT]{user terminal}
\acro{bs}[BS]{base station}
\acro{sm}[SM]{spatial modulation}
\acro{masm}[MA-SM]{multiple-active spatial modulation}
\acro{glse}[GLSE]{generalized least squared error}
\acro{rls}[RLS]{regularized least-squares}
\acro{rhs}[r.h.s.]{right hand side}
\acro{lhs}[l.h.s.]{left hand side}
\acro{wrt}[w.r.t.]{with respect to}
\acro{rs}[RS]{replica symmetry}
\acro{rsb}[RSB]{replica symmetry breaking}
\acro{ml}[ML]{maximum-likelihood}
\acro{rzf}[RZF]{regularized zero forcing}
\acro{psk}[PSK]{phase shift keying}
\acro{snr}[SNR]{signal-to-noise ratio}
\acro{rf}[RF]{radio frequency}
\acro{tdd}[TDD]{time division duplexing}
\acro{mf}[MF]{match filtering}
\acro{gamp}[GAMP]{generalized approximate message passing}
\acro{map}[MAP]{maximum-a-posterior}
\acro{mse}[MSE]{mean squared error}
\acro{mmse}[MMSE]{minimum MSE}
\acro{ssk}[SSK]{space shift keying}
\acro{rss}[RSS]{residual sum of squares}
\acro{lasso}[LASSO]{least absolute shrinkage and selection operator}
\end{acronym}

\title{RLS-Based Detection for Massive Spatial Modulation MIMO}
\author{
\IEEEauthorblockN{
Ali Bereyhi,
Saba Asaad,
Bernhard G\"ade,
Ralf R. M\"uller
}
\IEEEauthorblockA{
Institute for Digital Communications, Friedrich-Alexander Universit\"at Erlangen-N\"urnberg\\
\{ali.bereyhi, saba.asaad, bernhard.gaede, ralf.r.mueller\}@fau.de
\thanks{This work has been accepted for presentation in the IEEE International Symposium on Information Theory (ISIT) 2019 in Paris, France. The link to the final version in the Proceedings of ISIT will be available later.}
}
}
\IEEEoverridecommandlockouts

\maketitle

\begin{abstract}
Most detection algorithms in spatial modulation (SM) are formulated as linear regression via the regularized least-squares (RLS) method. In this method, the transmit signal is estimated by minimizing the residual sum of squares penalized with some regularization. This paper studies the asymptotic performance of a generic RLS-based detection algorithm employed for recovery of SM signals. We derive analytically the asymptotic average mean squared error and the error rate for the class of \textit{bi-unitarily invariant} channel matrices. 

The~analytic~results~are employed to study the performance of SM detection via the box-LASSO. The analysis demonstrates that the performance characterization for i.i.d. Gaussian channel matrices is valid for matrices with non-Gaussian entries, as well. This justifies the partially approved conjecture given in \cite{atitallah2017box}. The derivations further extend the former studies to scenarios with non-i.i.d. channel matrices.~Numerical~investigations validate the analysis, even for practical system dimensions. 
\end{abstract}

\IEEEpeerreviewmaketitle
\section{Introduction}
\label{sec:intro}
Recovery algorithms based on the \ac{rls} regression method are doubtless among the most studied schemes for ill-posed signal recovery in linear models. Examples of such algorithms are the \ac{lasso} \cite{tibshirani1996regression} and Tikhonov regularization method \cite{tikhonov1977methods} which are used in various applications,~e.g. sparse recovery \cite{candes2006robust}. The main task in these applications~is~to recover a signal $\bx$ from an underdetermined set of linear and~possibly~noisy observations $\by$. The \ac{rls}-based algorithms solve this problem by minimizing the \ac{rss} penalized via a regularization term, i.e. $\norm{\mH\bx-\by}^2+ f\brc{\bx}$ for some penalty $f\brc\cdot$ and the known projection~$\mH$.

This study investigates the characteristics of a generic class of \ac{rls}-based algorithms by considering~their applications~to the \ac{sm} technique recently proposed~for \ac{mimo} transmission~with~restricted hardware complexity \cite{jeganathan2009space,younis2010generalised}. For scenarios in which the number of active transmit antennas is smaller than the total number of available antenna elements, \ac{sm} utilizes the sparsity of the transmit signal to convey information via both the non-zero symbols and the support of the transmit signal. To this end, the information bits are divided into two subsets. One subset is used to select a unique subset of transmit antennas which are set active during transmission, and the other subset is transmitted over these active antennas via a conventional modulation technique, e.g. \ac{psk}. At the receive side, the receiver needs to jointly recover the index of the active antennas, and the transmitted symbols. 

Following analogy between the detection task in \ac{sm} and compressive sensing \cite{candes2006robust,donoho2006compressed}, several detection schemes were proposed in the literature based on sparse recovery techniques, e.g. \cite{yu2012compressed,xiao2017efficient,atitallah2017box}. These schemes often lie in the class of \ac{rls}-based recovery algorithms; hence, their performance~is~characterized by studying these recovery algorithms.
\subsection*{Contributions and Related Work}
Recent developments in multi-antenna technologies suggest the utilization of massive \ac{mimo} settings in the next generation of mobile networks \cite{hoydis2013massive}. In this respect, we characterize the \ac{rls}-based recovery algorithms in the large-system limit. The analysis differs from earlier studies on mathematically similar models, e.g. \cite{bereyhi2016statistical}, in the fact that the \textit{\ac{sm} signals are not \ac{iid}} in general. We address this via a conditional form of the \textit{asymmetric~decoupling property} derived for \ac{rls} recovery in \cite{bereyhi2018MAP}. Closed-form expressions for the average \ac{mse} and error rate in a massive multiuser \ac{masm} \ac{mimo} system are derived.

Our derivations extend available results in the literature in various respects. We demonstrate this by studying the example of~box-\ac{lasso}~recovery. For this example, our results extend the earlier derivations for \ac{iid} Gaussian channel matrices in \cite{atitallah2017box} to the class of \textit{bi-unitarily invariant} random matrices. Our investigations further show that the asymptotic distortion is of the same form for all \ac{iid}~channel~matrices. This justifies the universality conjecture which was partially approved in \cite{atitallah2017box}. 

%
%
%
\subsection*{Notations} 
We represent scalars, vectors and matrices with non-bold, bold lower case and bold upper case letters, respectively. $\mI_K$ denotes the $K \times K$ identity matrix. $\mH^{\trp}$~and~$\mH^{\her}$ are the transposed and transposed conjugate of $\mH$. The real axes~and its non-negative subset are denoted by $\setR$ and $\setR^+_0$, respectively. The complex plane is shown by $\setC$. We use $\norm{\cdot}$ and $\norm{\cdot}_1$ to show the $\ell_2$- and $\ell_1$-norm, respectively. $\norm{\bx}_0$ denotes the ``$\ell_0$-norm'' of $\bx$ defined as the number of non-zero entries. For random variable $x$, the probability distribution is shown by $\rmp\brc{x}$. $\Ex{ \cdot}{}$~is~the expectation operator. We use the shortened notation $[N]$ to represent $\set{1, \ldots , N}$.
\section{Problem Formulation}
\label{sec:sys}
Consider uplink transmission in a Gaussian multiple access \ac{mimo} channel with $K$ transmitters and a single receiver. Each transmitter is equipped with $M_{\rm u}$ antennas and $L_{\rm u}$ \ac{rf} chains. Hence, in each channel use, only $L_{\rm u}$ transmit antennas are active at transmit terminals. We denote the fraction of active antennas by $\eta = {L_{\rm u}}/{M_{\rm u}}$.

At the receive side, an antenna array of size $N$ is employed. Hence, the receive vector $\by\in\setC^{N}$ reads
\begin{align}
\by = \left. \mH \right. \bx + \bn. \label{eq:sys}
\end{align}
Here, 
$\bx$ represents the vector of transmit signals, i.e. 
\vspace*{-1mm}
\begin{align}
\bx = \left[ \bx_1^\trp,\ldots,\bx_K^\trp \right]^\trp
\end{align}
where $\bx_k\in\setC^{M_{\rm u}}$ is the signal transmitted by terminal $k$. As a result, $\bx$ is of size $M = K M_{\rm u}$. $\mH\in\setC^{M\times N}$ denotes the channel matrix. It is assumed that~the~channel~is~estimated~prior to data transmission, and hence $\mH$ is known at the receiver.\\ $\bn\in \setC^{N}$ represents noise and reads $\bn \sim\mathcal{CN}\brc{\boldsymbol{0}, \sigma^2 \mI_N}$.

\vspace*{-1mm}
\subsection{Channel Model}
The channel is assumed to experience frequency-flat fading with slow time variations. 
We consider a generic fading model in which $\mH$ is a \textit{bi-unitarily invariant} random matrix. This means that for any pair of unitary matrices $\mU\in\setC^{N\times N}$ and $\mV\in\setC^{M\times M}$, which~are~independent of $\mH$,~the~joint~distribution of the entries of $\mH$ is identical to that of $\mU \mH \mV$~\cite{tulino2004random}.

 This ensemble comprises~a variety of fading models including the well-known \ac{iid} Rayleigh fading model.
 
\vspace*{-1mm}
\subsection{Spatial Modulation}
%
Let $\bdd_k$ denote a  sequence of information bits at terminal $k$. Without loss of generality, let $\bdd_k$ be an \ac{iid} binary sequence with uniform distribution. Assume that the active antennas~at each terminal transmit symbols from~alphabet~$\setS$~which~contains $2^S$ symbols. $\bx_k$ is then constructed as follows:
\begin{enumerate}
\item \textit{Assigning modulation indices:} Consider all possible subsets of $L_{\rm u}$ transmit antennas selected out of $M_{\rm u}$ available antennas at terminal $k$. The transmitter selects $2^I$ distinct subsets randomly and uniformly, where
\begin{align}
I = \left\lfloor \log_2 {M_{\rm u}\choose L_{\rm u}} \right\rfloor.
\end{align}
To each of these subsets a \textit{modulation index} is assigned.
\item \textit{Index modulation:} For given sequence $\bdd_k$, the transmitter chooses index $i_k\in [ 2^I ]$, such that the first $I$ bits of $\bdd_k$ be the binary representation of $i_k$. We denote the subset of $L_{\rm u}$ antenna indices corresponding to $i_k$ with $\setL \brc{i_k}$. 
\item \textit{Modulating multiple streams:} Terminal $k$ considers $L_{\rm u}$ independent blocks of $\bdd_k$, each of length $S$, and maps them into the symbols $s_{k}\brc{m} \in \setS$ for $m\in\setL\brc{i_k}$, using a standard modulation scheme, e.g. \ac{psk}. These symbols are then transmitted on the antennas corresponding to $i_k$.
\end{enumerate}

Considering the above \ac{sm} scheme, the $m$-th transmit signal entry of terminal $k$, i.e.~$x_{k,m}$~for $m\in [M_{\rm u}]$, reads $x_{k,m}= s_{k}\brc{m} \mone\set{m\in \setL\brc{i_k}}$, where $\mone\set{m\in \setL\brc{i_k}}$ is the indicator function, returning one, if $m\in \setL\brc{i_k}$ and zero, if $m\notin \setL\brc{i_k}$. %
Hence, $\bx_k$ is an $L_{\rm u}$-sparse vector, i.e. only $L_{\rm u}$ entries are non-zero. As a result, $\bx$ is $L$-sparse, where $L=K L_{\rm u}$. 

We define the \textit{activity factor} as 
$\mathrm{AF} \coloneqq {\norm{\bx}_0}/{M}$. %
This factor describes the total fraction of active transmit antennas in the network. Noting that $M=KM_{\rm u}$, we have $\mathrm{AF} = \eta$.

\subsection{RLS-Based Detection Algorithms}
For data recovery, the receiver requires to detect both the support of $\bx$ and the transmitted symbols from $\by$. To~this~end, an \ac{rls}-based algorithm is employed. This algorithm determines a \textit{soft estimation} of $\bx$, for given $\mH$, as 
\begin{align}
\bx^\star \coloneqq \argmin_{\bv\in\setX_0^M} \left. \norm{\by-\mH \bv}^2 + \right. f_\reg \brc{\bv} ,
\end{align}
where $f_\reg \brc\cdot$ is a regularization function, and $\setX$ is a subset of $\setC$ including the alphabet $\setS$, i.e. $\setS\subseteq \setX \subseteq \setC$. The notation $\setX_0 \coloneqq \set{0}\cup\setX$ is further defined for brevity. 

Given $\bx^\star$, the detected vector is then given by mapping the soft estimation into a vector in $\setS_0^M$. This means $\hat{\bx} = f_\dec \brc{\bx^\star}$ where $f_{\rm dec} \brc{\cdot}: \setX_0^M \mapsto \setS_0^M$ is a decisioning function.

\subsubsection*{Special Cases}
The \ac{rls}-based recovery schemes recover various \ac{masm} detection algorithms. As an example, let~$\setX=\setS$ and $f_\reg\brc{\bx} = -\sigma^2 \loge \rmp\brc\bv$ with $\rmp\brc\cdot$ denoting the prior distribution of $\bx$. Moreover, set $f_\dec\brc{\bx} = \bx$. The algorithm in this case reduces to the optimal Bayesian detection scheme, i.e. \textit{\ac{map}}. 

Noting that dimensions are rather large in massive \ac{mimo} systems, the non-convex choices of $\setX_0$ and $f_\reg\brc{\bx}$ result~in computationally intractable detection algorithms. As a result, convex forms are considered in practice. A well-known example is the so-called \textit{box-\ac{lasso}}. 
In box-\ac{lasso}, $\setX_0$ is set to a convex subset of the complex plane which contains the symbols in $\setS_0$. $f_\reg\brc{\bx}$ is further chosen proportional to the $\ell_1$-norm following the regularization approach proposed by Tibshirani in \cite{tibshirani1996regression}. In this case, the soft estimation comprises symbols from $\setX_0$. $f_{\rm dec} \brc{\cdot}$ is hence set to be an entry-wise thresholding operator which maps each estimated entry into either zero or a transmit symbol in $\setS$. 
\subsection{Performance Metrics}
Using the \ac{masm} scheme, each terminal transmits $I + L_{\rm u} S$ information bits per channel use. Comparing to conventional modulation schemes, the spectral efficiency in this~case~is~increased by $I$ bits per channel use. This is acquired~at~the~expense of diversity. In fact, the index modulation in \ac{masm} performs as a random selection algorithm which reduces the diversity gain, compared to other selection techniques. %

In order to characterize the performance of \ac{masm} transmission, a distortion metric is further required. The common metric is the \textit{average error rate} defined as the probability of bit~flip~averaged over the transmit block size, i.e. $M$. 
\begin{definition}[Average error rate]
Let $\hat{\bx}$ denote the vector of detected signals. The average error rate is defined as
\begin{align}
\bar{P}_{\rm E} \brc{M}  \coloneqq \frac{1}{M} \sum_{m=1}^M \Pr \set{\hat{x}_m\neq x_m}. \label{eq:R}
\end{align}
\end{definition}
For \ac{rls}-based detectors, the distortion can also be defined with respect to the soft estimation. In this respect, we further consider the average \ac{mse} 
as another distortion metric. 

\begin{definition}[Average \ac{mse}]
\label{dist}
Consider the soft estimation $\bx^\star$. The average \ac{mse} is defined as
\begin{align}
\mse \brc{M} \coloneqq   \frac{1}{M} \E \norm{\bx^\star - \bx_m}^2.  \label{eq:dist}
\end{align}
\end{definition}


\subsection{The Large-System Limit}
\label{sec:largelimit}
We aim to determine the asymptotic limit of the distortion metrics. In this respect, for bounded and fixed $L_{\rm u}$ and $M_{\rm u}$, we consider a sequence of settings with $N$ receive antennas and $K$ transmit terminals. We assume that $K$ is a deterministic sequence of $N$, such that
\begin{align}
\alpha = \lim_{N\uparrow\infty} \frac{K}{N}
\end{align}
is bounded. As a result, the sequence of \textit{channel loads}, defined as $\xi\coloneqq N/M$, converges to $\xi = {1}/{\alpha M_{\rm u}}$, as $N\to\infty$. The activity factor of this sequence of settings is constant in $N$ and reads $\mathrm{AF} = \eta$ for all choices of $N$. 

For a given setting with $N$ receive antennas, we define $\rmF_{\mJ}^N$ to be the empirical cumulative distribution of the eigenvalues of $\mJ = \mH^\her\mH$. It~is assumed that the sequence of $\rmF_{\mJ}^N$ has a deterministic limit~$\rmF_\mJ$, when $N$ grows large. The~Stieltjes transform of $\rmF_\mJ$ is given by
\begin{align}
\rmG \brc{z} = \int \frac{\dif \rmF_\mJ \brc{\lambda}}{\lambda - z} \label{eq:R-trans}
\end{align}
for some $z$ in the upper half complex plane. The $\rmR$-transform is further defined as $\rmR\brc{\omega} = \rmG^{-1} \brc{-\omega} - 1/\omega$ for~some~$\omega\in\setC$, such that $\rmR\brc{0} = \int \lambda \dif \rmF_\mJ \brc{\lambda}$. Here, $\rmG^{-1} \brc{\cdot}$ denotes the inverse with respect to composition.

\section{Main Results}
\label{sec:large}
The analytic derivations in this section follows the results given in \cite{bereyhi2018MAP} via the replica method. The consistency of the results relies on the validity of some conjectures,~such~as replica continuity and replica symmetry. 
Although~these~conjectures lack mathematical proofs, several studies have confirmed the consistency for some particular examples, e.g. \cite{reeves2016replica}.

We state the main results using the \textit{decoupled setting}.~This is a tunable scalar setting whose distortion metrics are analytically calculated for all tuning parameters. It is shown that for a specific tuple of the tuning parameters, the distortion metrics of this setting give the large-system limits~of the average error rate and \ac{mse}.

\begin{definition}[Decoupled setting]
\label{def:Dec}
For given $c$ and $q$, define 
\begin{subequations}
\begin{align}
\tau \brc{c} &\coloneqq\frac{1}{ \rmR \brc{-c} } \label{eq:rs-4a} \\
\theta\brc{c,q} &\coloneqq \frac{1}{ \rmR \brc{-c} } \sqrt{\frac{\partial}{\partial c} \left[ \brc{ \sigma^2 c - q } \rmR \brc{-c} \right]}. \label{eq:rs-4b}
\end{align}
\end{subequations}
Let $\xx = \psi \rms$, where $\rms$ is uniformly distributed on $\setS$, and~$\psi$ is a Bernoulli random variable for which $\Pr\set{\psi=1} = 1- \Pr\set{\psi=0} = \eta$. Define the decoupled output $\yy\brc{c,q}$ as 
\begin{align*}
\yy\brc{c,q} = \xx +  \theta \brc{c,q} \zz \label{eq:yy_dec}
\end{align*}
with $\zz\sim\mathcal{CN} \brc{0, 1 }$. Then, the decoupled soft estimation is 
\begin{equation}
\label{eq:dec_est}
\xx^\star\brc{c,q} \coloneqq \argmin_{\vv\left. \in\right. \setX_0} \left. \frac{1}{\tau\brc{c}}\right. \left. \abs{\yy\brc{c,q}- \vv}^2 + f_\reg \brc{\vv} \right. .
\end{equation}
The decoupled detected symbol is moreover defined in terms of $\xx^\star$ as
$\hat\xx \brc{c,q} \coloneqq f_{\dec} \brc{ \xx^\star\brc{c,q} }$. %
For this setting, the error probability reads $Q_{\rm E} \brc{c,q} = \Pr \set{ \hat\xx \brc{c,q} \neq \xx }$
and the \ac{mse} is given by %
$\Gamma\brc{c,q} = \Ex{ \abs{\xx^\star \brc{c,q} - \xx}^2 }{}$.
\end{definition}

Unlike $\bar{P}_{\rm E} \brc{M}$ and $\mse \brc{M}$, the derivation of $Q_{\rm E} \brc{c,q}$ and $\Gamma\brc{c,q}$ deals with a scalar optimization problem which is tractable for any $\setX$ and $f_\reg \brc\cdot$. Proposition~\ref{proposition:RS} indicates that for specific choices of $c$ and $q$, the average error rate and the average \ac{mse} are given by the error probability and the \ac{mse} of the decoupled setting, respectively. The values~of~$c$~and~$q$, for which this equivalency happens, is given in the proposition via a system of fixed-point equations.

\begin{proposition}
\label{proposition:RS}
Consider the sequence of settings illustrated in Section~\ref{sec:largelimit}. As $M$ grows large, we have
\begin{align}
\lim_{M\uparrow \infty} \bar{P}_{\rm E} \brc{M} = Q_{\rm E} \brc{c^\star,q^\star}
\end{align}
and
\begin{align}
\lim_{M\uparrow \infty} \mse \brc{M} = \Gamma\brc{c^\star,q^\star}
\end{align}
where $c^\star$ and $q^\star$ are solutions to the fixed-point equations
\begin{subequations}
\begin{align}
c \left. \theta\brc{c,q} \right.  &= \left. \tau\brc{c} \right. \Ex{ \real{\brc{\xx^\star\brc{c,q} - \xx } \zz^*} }{} \label{eq:rs-6a} \\
q &= \Ex{ \abs{\xx^\star\brc{c,q} - \xx}^2  }{}. \label{eq:rs-6b}
\end{align}
\end{subequations}
\end{proposition}
\begin{proof}
The proof follows the \textit{asymmetric} form of the \textit{decoupling property} derived for \ac{rls} recovery in \cite{bereyhi2018MAP}. To start with the proof, let $\bi = \left[ i_1 , \ldots, i_K\right]^\trp$  be the vector of modulation indices for a setting with transmit dimension $M$. For a given distortion function $F_{\rm D} \brc{\cdot;\cdot}: \setX_0 \times \setS_0 \mapsto \setR^+_0$, define the conditional average distortion as
\begin{align}
D_M\brc{\bi} &\coloneqq \frac{1}{M} \sum_{m=1}^M \Ex{ F_{\rm D} \brc{x^\star_m;x_m} \vert \bi }{}. \label{eq:D_i}
\end{align}
For a given $\bi$, let 
$\Sp\bi = \set{ m \in [M] : x_m \neq 0 }$ 
contain the indices of the non-zero entries in $\bx$. It~is~then~straightforward to conclude that conditioned to $\bi$, the entries $x_m$ for $m\in\Sp\bi$ are independent and uniformly distributed~on~$\setS$. As a result, the conditional distribution $\rmp\brc{\bx|\bi}$ reads
\begin{align}
\rmp\brc{\bx|\bi} &= \prod_{m\left.\in\right.\Sp\bi} 2^{-S} \prod_{m\left.\notin\right.\Sp\bi} \mone\set{ x_m = 0}. \label{eq:condDist}
\end{align}
This distribution follows the asymmetric signal model that considered in \cite{bereyhi2018MAP} with two blocks of \ac{iid} sequences. 

By standard derivations, it is concluded from the asymptotic decoupling property in \cite{bereyhi2018MAP} that for $m\in\Sp\bi$ and $m\notin\Sp\bi$, $\rmp\brc{x_m^\star\vert x_m, \bi}$ converges to $\rmp\brc{\xx^\star\brc{c^\star,q^\star}\vert \xx=\rms}$ and $\rmp\brc{\xx^\star\brc{c^\star,q^\star}\vert \xx=0}$, respectively, as $N$ tends to $\infty$. Here, $\rmp\brc{\xx^\star\brc{c^\star,q^\star}\vert \xx=u}$ denotes the conditional distribution of the~decoupled~soft~estimation defined in \eqref{eq:dec_est}, for $c=c^\star$ and $q=q^\star$, given that the decoupled input is set to $\xx=u$. $\rms$ is moreover distributed uniformly over $\setS$. Substituting into \eqref{eq:D_i}, after some lines of calculations, we have
%
%
\begin{align}
\hspace{-3mm}\lim_{M\uparrow\infty} \hspace{-1mm} D_M\brc{\bi} \hspace{-.5mm} = \hspace{-.5mm} &\left. \eta \right. \Ex{\Ex{F_{\rm D} \brc{\xx^\star\brc{c^\star,q^\star} ;\rms} \vert\xx=\rms }{ } }{\rms} \nonumber \\
&+ \brc{1-\eta} \Ex{F_{\rm D} \brc{\xx^\star\brc{c^\star,q^\star};0} \vert \xx=0}{ }.
\label{eq:LimitDISTortion}
\end{align}
Noting that the limit in \eqref{eq:LimitDISTortion} is constant in $\bi$, we finally have
\begin{align}
\lim_{M\uparrow\infty} D_M = \lim_{M\uparrow\infty} D_M\brc{\bi},
\end{align}
where $D_M \coloneqq \Ex{D_M\brc{\bi}}{\bi}$.

By setting the distortion function to $F_{\rm D} \brc{x^\star ; x} = \abs{x^\star - x}^2$ and $F_{\rm D} \brc{x^\star ; x} = \mone\set{f_\dec\brc{x^\star} \neq x}$, the asymptotic average \ac{mse} and error rate are derived from \eqref{eq:LimitDISTortion}, respectively. More details are given in the extended version of the paper.
\end{proof}

\section{A Classic Application of the Results}
The asymptotic characterization of the \ac{rls}-based detection schemes enables us investigating \ac{masm} in various aspects. In this section, we discuss a particular application;~namely,~we study a \textit{box-\ac{lasso}} detector. This special form was investigated earlier in \cite{atitallah2017box}, where its asymptotic performance was characterized analytically for \ac{iid} Gaussian channel matrices. The analysis in this section hence extends the derivations in \cite{atitallah2017box} in several aspects, e.g. channel model. 

Applications of the results are not restricted to this particular example. For instance, the asymptotic characterizations can be utilized to derive theoretical bounds on the performance of \ac{masm} transmission. Due to page limitation, we skip further applications and present them later in an extended version. 

\subsection{Box-LASSO Detection for SSK Transmission}
We consider the following special case of the setting: 
\begin{enumerate}
\item $\setS = \{\sqrt{P} \}$ for some power $P$.
This equivalently means $S=0$. This is an extreme case of \ac{sm}, known as \ac{ssk}, in which the~information~symbols~are completely conveyed via index modulation.
\item The \ac{rls}-based detector has the following specifications:
\begin{itemize}
\item The regularization function is set to $f_{\reg} \brc{\bv} = \lambda \norm{\bv}_1$ for some \textit{regularization parameter} $\lambda$ which is tunable. 
\item $\setX = \left[ -\ell , u \right]$ for some $\ell \geq 0$ and $u \geq \sqrt{P}$.
\item The decisioning function is $f_{\dec} \brc{x}  = \left. \sqrt{P} \right. \mone\set{x\geq \epsilon}$ for some \textit{threshold} $\epsilon$.
\end{itemize}
\end{enumerate}

This setting describes box-\ac{lasso} detection for \ac{ssk} signaling. The detector, in this case, relaxes the optimal Bayesian detector by approximating the exponent of the transmit signal's prior distribution with the $\ell_1$-norm and convexifying the set $\setS_0 =\{0,\sqrt{P}\}$ with the interval $\setX_0 = \left[ -\ell , u \right]$.

The asymptotic performance of this particular setting was investigated in \cite{atitallah2017box} for \ac{iid} Gaussian channel matrices. Based on simulations and universality results\footnote{See \cite{oymak2017universality} for some results on universality.}, the authors~conjec-tured that the analysis is further valid for \textit{non-Gaussian}~\ac{iid} channel~matrices. This conjecture was partially approved in \cite{atitallah2017box} using the Lindeberg principle.

Invoking our asymptotic derivations, the conjecture in \cite{atitallah2017box} is straightforwardly justified. In fact, from Proposition~\ref{proposition:RS}, it is observed that the channel matrix plays rule in the asymptotic characterization via the limiting distribution $\rmF_\mJ$. From random matrix theory, we know that this distribution is the same~for~all \ac{iid} channel matrices and follows the Marcenko-Pastur law \cite{tulino2004random}. It is hence concluded that the asymptotic characterizations of the performance for \ac{iid} Gaussian matrices extends to \ac{iid} matrices with other entry distributions, as well. 

\begin{remark}
\normalfont
The asymptotic results in this paper are given for bi-unitarily invariant random matrices. Hence, the analysis not only justifies the conjecture in \cite{atitallah2017box}, but also~extends~the~results beyond the \ac{iid} matrices. For non-\ac{iid} matrices, it is obvious from Proposition~\ref{proposition:RS}, that the derivations in \cite{atitallah2017box} are not valid anymore. The performance in such cases is straightforwardly characterized via Proposition~\ref{proposition:RS}.
\end{remark}
\subsection{Decoupled Setting of the Box-LASSO Detector}
For the box-LASSO detector, the decoupled input reads $\xx = \sqrt{P} \psi $ with $\psi$ being a Bernoulli random variable described in Definition~\ref{def:Dec}. The decoupled soft estimation is further given~by 
\begin{align}
\xx^\star\brc{c,q} \hspace*{-1mm} = \hspace*{-1mm}
\begin{cases}
u & \phantom{-\theta- \ell  \leq} \ \yy\brc{c,q} \geq \theta   + u \vspace*{2mm}\\
\yy\brc{c,q} - \theta & \phantom{- - \ell}\theta  \leq \yy\brc{c,q} \leq \theta + u\vspace*{2mm}\\
0 & \phantom{-\ell} -\theta  \leq \yy\brc{c,q} \leq \theta \vspace*{2mm}\\
\yy\brc{c,q} + \theta & -\theta-\ell  \leq \yy\brc{c,q} \leq -\theta\vspace*{2mm}\\
-\ell & \phantom{-\theta-\ell  \leq} \ \yy\brc{c,q} \leq -\theta - \ell
\end{cases}\label{eq:dec_box_lasso}
\end{align}
where $\theta \coloneqq {\tau\brc{c}\lambda}/{2}$ and $\yy\brc{c,q} = \sqrt{P} \psi +\theta\brc{c,q} \zz$.

The asymptotic values for the \ac{mse} and error rate~are~derived by substituting \eqref{eq:dec_box_lasso} into the fixed point equations, given in Proposition~\ref{proposition:RS}, and calculating $c^\star$ and $q^\star$. 

Noting that $\zz$ is a zero-mean unit-variance complex Gaussian random variable, the expectations on the right hand side of \eqref{eq:rs-6a} and \eqref{eq:rs-6b} are analytically derived for a given~pair~of~$c$ and $q$ as sums of Gaussian integrals. Hence, they are straightforwardly calculated and replaced into the fixed-point equations. The resulting equations are then solved numerically.\footnote{An alternative approach is to iteratively find the stability point of~the~corresponding \textit{replica simulator}; see \cite{bereyhi2016statistical} for detailed discussions.}

\subsection{Numerical Results}
The analytic derivations are validated via numerical investigations considering the example~of~box-LASSO~detection. To this end, we consider a scenario with~$K=20$~transmit~terminals, each equipped with $M_{\rm u} = 8$ antennas and a single \ac{rf} chain, i.e. $L_{\rm u} =1$. Consequently, $\eta = L_{\rm u} / M_{\rm u} = 0.125$. Using the single \ac{rf}-chain, each terminal transmits three information bits in each channel use by \ac{ssk} signaling. We also set $P=1$.

At the receiver, an antenna array of size $N=80$~is~employed. Thus, $\alpha = K/N = 0.25$, and the channel load~is~$\xi = 1/M_{\rm u} \alpha = 0.5$. For signal recovery, a box-\ac{lasso} detector is used in which $\ell =0$ and $u =\sqrt{P} = 1$. The threshold in the decisioning function is set to $\epsilon = 0.5$.

The analytical results are given for an \ac{iid} channel matrix whose entries are zero-mean with variance $1/M$.\footnote{The entries can have an arbitrary distribution.} In this case, we have $\rmR\brc{c} = {\xi} / \brc{1-c}$ \cite{tulino2004random}. The simulations are given for Rayleigh fading model in which the channel entries are Gaussian. We set the noise variance with respect to the \ac{snr} which is defined as $\snr = P/\sigma^2$.

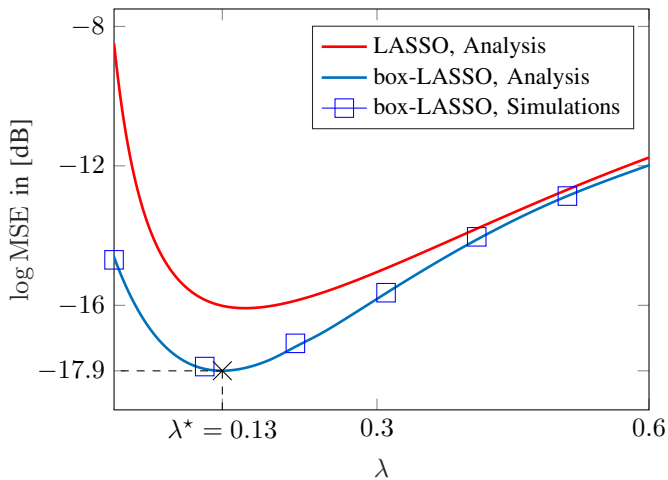
\begin{figure}[t]
\centering
\input{Figs/MSEvsLambda.tex}
\caption{\ac{mse} vs. $\lambda$. The simulations closely track analytical results. For box-\ac{lasso} detection, the \ac{mse} is minimized at $\lambda^\star = 0.13$.}
\label{fig:1}
\end{figure}

Fig.~\ref{fig:1} and \ref{fig:2} depict the average \ac{mse} and the error rate, achieved by the box-\ac{lasso} detector, against regularization parameter $\lambda$, when $\log \snr\hspace*{-.5mm} = \hspace*{-.5mm} 14$ dB. For sake of comparison, the plots for \textit{standard \ac{lasso}} are further given. By standard \ac{lasso}, we mean the extreme case of box-\ac{lasso} detection in which $u,\ell \uparrow \infty$. The solid lines in these figures indicate the analytic results given by Proposition~\ref{proposition:RS}. The squares~are given by numerical simulations with $1000$ realizations. As~the figures depict the simulated points closely track the large-system characterization. This observation validates the consistency of the analytic results in practical dimensions.

As Fig.~\ref{fig:1}~and~\ref{fig:2} demonstrate, for a given \ac{snr},~the~performance is optimized at some $\lambda^\star$. This value is analytically found via Proposition~\ref{proposition:RS}, as Fig.~\ref{fig:3} illustrates. In this figure,~$\lambda^\star$ is plotted against the \ac{snr}. At each \ac{snr}, $\lambda^\star$ is found such that the asymptotic \ac{mse}, derived via Proposition~\ref{proposition:RS}, is minimized. The regularization parameter in the box-\ac{lasso} detector is then set to $\lambda^\star$ and the achieved \ac{mse} is plotted against the \ac{snr}. As the figure depicts, numerical simulations are tightly consistent with the analytical derivations. This implies that the analytic derivations provide an easy and fast tool for efficient tuning of \ac{rls}-based detectors, even in practical dimensions.

\section{Conclusion}
The average \ac{mse} and error rate were analytically derived for massive \ac{sm} \ac{mimo} settings when a generic \ac{rls}-based algorithm is employed for detection. The analysis was given for the large class of bi-unitarily invariant random matrices which includes various fading models. The asymptotic results were employed to study box-\ac{lasso} detectors for \ac{ssk} signaling. This particular application of the results extended the analysis in \cite{atitallah2017box} to a larger set of channel matrices. Numerical simulations showed close consistency with the~analytic~results, even in practical dimensions.~Other applications of the results are skipped due to the page~limitation and will be presented~in~an extended version which is currently under preparation.

\begin{figure}[t]
\centering
\input{Figs/PEvsLambda.tex}
\caption{Average error rate vs. $\lambda$ considering both analytical results, given by Proposition~\ref{proposition:RS}, and numerical simulations.}
\label{fig:2}
\end{figure}
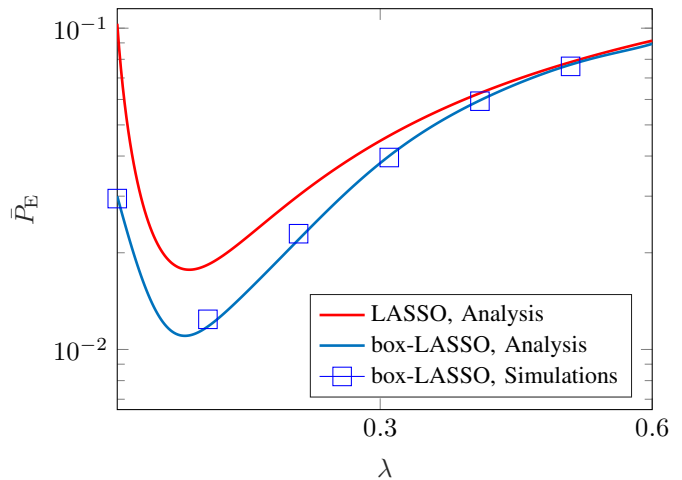

\begin{figure}[t]
\centering
\input{Figs/LambdaSTARvsSNR.tex}
\caption{Optimal regularization parameter and minimum achievable \ac{mse} via box-\ac{lasso} detection vs. \ac{snr}. The curves match the result seen in Fig.~\ref{fig:1}.}
\label{fig:3}
\end{figure}
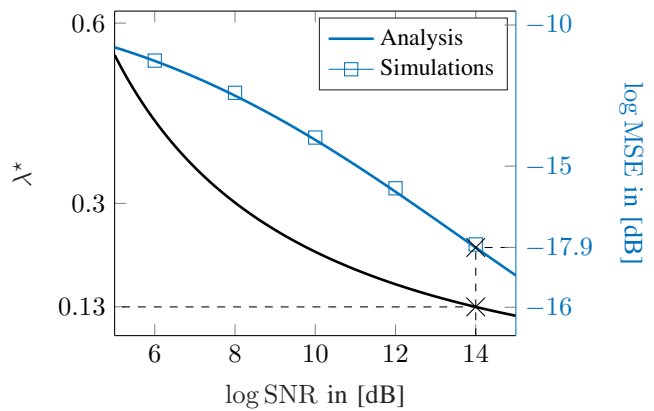

\bibliography{ref}
\bibliographystyle{IEEEtran}
\end{document}

%% file: commands.tex
\newcommand{\setX}{\mathbbmss{X}}
\newcommand{\setL}{\mathbbmss{L}}

\newcommand{\setR}{\mathbbmss{R}}

\newcommand{\setS}{\mathbbmss{S}}

\newcommand{\setC}{\mathbbmss{C}}

\DeclareMathOperator*{\argmin}{argmin}

\newcommand{\rmp}{\mathrm{p}}

\newcommand{\rmF}{\mathrm{F}}

\newcommand{\rmR}{\mathrm{R}}

\newcommand{\rmG}{\mathrm{G}}

\newcommand{\her}{\mathsf{H}}

\newcommand{\bi}{\mathbf{i}}

\newcommand{\rms}{\mathrm{s}}

\newcommand{\dec}{\mathrm{dec}}
\newcommand{\reg}{\mathrm{reg}}
\newcommand{\loge}{\left. \mathrm{ln} \right.}

\newcommand{\bx}{{\boldsymbol{x}}}

\newcommand{\vv}{\mathrm{v}}

\newcommand{\xx}{\mathrm{x}}
\newcommand{\yy}{\mathrm{y}}
\newcommand{\zz}{\mathrm{z}}

\newcommand{\bv}{{\boldsymbol{v}}}

\newcommand{\bdd}{{\mathbf{d}}}

\newcommand{\dif}{\mathrm{d}}

\newcommand{\by}{{\boldsymbol{y}}}

\newcommand{\bn}{{\boldsymbol{n}}}

\newcommand{\trp}{\mathsf{T}}

\newcommand{\mI}{\mathbf{I}}
\newcommand{\mone}{\mathbf{1}}
\newcommand{\mJ}{\mathbf{J}}

\newcommand{\mU}{\mathbf{U}}

\newcommand{\mV}{\mathbf{V}}

\newcommand{\mH}{\mathbf{H}}

\newcommand{\snr}{\mathrm{SNR}}

\newcommand{\E}{\mathbbmss{E}\hspace{.5mm}}

\newcommand{\mse}{\mathrm{MSE}}

\newcommand{\Ex}[2]{\mathbbmss{E}_{#2}\left\{#1\right\}}
\newcommand{\norm}[1]{\lVert #1 \rVert}
\newcommand{\set}[1]{\left\lbrace #1 \right\rbrace}
\newcommand{\brc}[1]{\left( #1 \right)}
\newcommand{\real}[1]{\mathbbm{Re}\left\lbrace #1 \right\rbrace}

\newcommand{\abs}[1]{\lvert #1 \rvert}

\newcommand{\Sp}[1]{\mathrm{Supp}\left( #1 \right) }

\newtheoremstyle{mystyle}
  {}
  {}
  {\it}
  {}
  {\bfseries}
  {:}
  { }
  {}

\theoremstyle{mystyle}

\newtheorem{definition}{Definition}
\newtheorem{proposition}{Proposition}

\newtheorem{remark}{Remark}

\algnewcommand\algorithmicLet{\textbf{Let}}
\algnewcommand\Let{\item[\algorithmicLet]}
\algnewcommand\algorithmicSet{\textbf{Set}}
\algnewcommand\Set{\item[\algorithmicSet]}

\algnewcommand\algorithmicInitiate{\textbf{Initiate}}
\algnewcommand\Initiate{\item[\algorithmicInitiate]}
\algnewcommand\algorithmicStart{\textbf{Begin}}
\algnewcommand\Begin{\item[\algorithmicStart]}
\algnewcommand\algorithmicEnd{\textbf{End}}
\algnewcommand\End{\item[\algorithmicEnd]}

\algnewcommand\algorithmicOutP{\textbf{Output:}}
\algnewcommand\Out{\item[\algorithmicOutP]}

\newcounter{bar}



%% file: Figs/MSEvsLambda.tex
%
%
\definecolor{mycolor1}{rgb}{0.00000,0.44700,0.74100}%
\begin{tikzpicture}
\begin{axis}[%
width=2.8in,
height=2.1in,
at={(1.989in,1.234in)},
scale only axis,
xmin=0.01,
xmax=0.6,
xtick={0.1294,0.3,0.6},
xticklabels={{$\lambda^\star=0.13$},{$0.3$},{$0.6$}},
xlabel style={font=\color{white!15!black}},
xlabel={$\lambda$},
ymin=-19,
ymax=-7.5,
ytick={-8,-12,-16,-17.88},
yticklabels={{$-8$},{$-12$},{$-16$},{$-17.9$}},
ylabel style={font=\color{white!15!black}},
ylabel={$\log \mse$ in [dB]},
axis background/.style={fill=white},
legend style={at={(0.37,0.7)}, anchor=south west, legend cell align=left, align=left, draw=white!15!black}
]

\addplot [color=red, line width=1.0pt]
  table[row sep=crcr]{%
0.01	-8.48616946962731\\
0.01	-8.48616946962731\\
0.0120205479452055	-8.97598331762338\\
0.014041095890411	-9.42769334389043\\
0.0160616438356164	-9.84526033467733\\
0.0180821917808219	-10.232114988308\\
0.02	-10.5736096588231\\
0.0201027397260274	-10.5912437330864\\
0.0221232876712329	-10.9252583959164\\
0.0241438356164384	-11.2364531593052\\
0.0261643835616438	-11.5268514327904\\
0.0281849315068493	-11.7982446626026\\
0.03	-12.0271492731691\\
0.0302054794520548	-12.0522246521861\\
0.0322260273972603	-12.2902106251071\\
0.0342465753424658	-12.5134720018183\\
0.0362671232876712	-12.7231476618775\\
0.0382876712328767	-12.9202623084509\\
0.04	-13.0781650761477\\
0.0403082191780822	-13.1057404312361\\
0.0423287671232877	-13.2804182691719\\
0.0443493150684931	-13.4450540994327\\
0.0463698630136986	-13.6003371196806\\
0.0483904109589041	-13.7468951429686\\
0.05	-13.8577851741602\\
0.0504109589041096	-13.8853012864319\\
0.0524315068493151	-14.0160798040009\\
0.0544520547945205	-14.1397111882658\\
0.056472602739726	-14.2566366461493\\
0.0584931506849315	-14.367262036255\\
0.06	-14.4458851557366\\
0.060513698630137	-14.4719613419496\\
0.0625342465753425	-14.571079742814\\
0.0645547945205479	-14.6649363376265\\
0.0665753424657534	-14.7538265641454\\
0.0685958904109589	-14.838024354362\\
0.07	-14.8939065490835\\
0.0706164383561644	-14.9177840583599\\
0.0726369863013699	-14.993342165253\\
0.0746575342465753	-15.0649188457432\\
0.0766780821917808	-15.1327193375031\\
0.0786986301369863	-15.1969351917548\\
0.08	-15.2364806967297\\
0.0807191780821918	-15.2577453970051\\
0.0827397260273973	-15.3153173938316\\
0.0847602739726027	-15.369807992847\\
0.0867808219178082	-15.421364206452\\
0.0888013698630137	-15.4701240036754\\
0.09	-15.4977810906776\\
0.0908219178082192	-15.5162169962761\\
0.0928424657534246	-15.5597650632979\\
0.0948630136986301	-15.6008829204258\\
0.0968835616438356	-15.6396786397524\\
0.0989041095890411	-15.6762541249208\\
0.1	-15.6951973274202\\
0.100924657534247	-15.710705546052\\
0.102945205479452	-15.7431237383717\\
0.104965753424658	-15.7735945680239\\
0.106986301369863	-15.8021992681788\\
0.109006849315068	-15.8290147482134\\
0.11	-15.8415617126376\\
0.111027397260274	-15.8541138784473\\
0.113047945205479	-15.8775657526602\\
0.115068493150685	-15.8994359303871\\
0.11708904109589	-15.919786660787\\
0.119109589041096	-15.9386770896978\\
0.12	-15.9465527543778\\
0.121130136986301	-15.9561634513349\\
0.123150684931507	-15.9722992459434\\
0.125171232876712	-15.9871354045936\\
0.127191780821918	-16.0007204421885\\
0.129212328767123	-16.0131005996599\\
0.13	-16.0176099460806\\
0.131232876712329	-16.0243199762301\\
0.133253424657534	-16.0344206525428\\
0.13527397260274	-16.0434428053879\\
0.137294520547945	-16.051424814685\\
0.139315068493151	-16.0584033633247\\
0.14	-16.0605474934711\\
0.141335616438356	-16.0644135304224\\
0.143356164383562	-16.0694888784814\\
0.145376712328767	-16.0736615349293\\
0.147397260273973	-16.0769622684428\\
0.149417808219178	-16.0794205604485\\
0.15	-16.0799765724002\\
0.151438356164384	-16.0810646721504\\
0.153458904109589	-16.0819217074072\\
0.155479452054795	-16.0820176717572\\
0.1575	-16.0813775278622\\
0.159520547945205	-16.0800252476228\\
0.16	-16.0796023684254\\
0.161541095890411	-16.0779838611941\\
0.163561643835616	-16.0752755031169\\
0.165582191780822	-16.0719214557593\\
0.167602739726027	-16.06794219025\\
0.169623287671233	-16.0633574050712\\
0.17	-16.0624371840912\\
0.171643835616438	-16.0581860624655\\
0.173664383561644	-16.0524464228\\
0.175684931506849	-16.0461560770197\\
0.177705479452055	-16.0393319773146\\
0.17972602739726	-16.0319904661118\\
0.18	-16.0309560422282\\
0.181746575342466	-16.0241473035001\\
0.183767123287671	-16.0158176931855\\
0.185787671232877	-16.0070163070669\\
0.187808219178082	-15.9977573085181\\
0.189828767123288	-15.9880543744548\\
0.19	-15.9872121122377\\
0.191849315068493	-15.9779207162598\\
0.193869863013699	-15.9673690996347\\
0.195890410958904	-15.9564118634418\\
0.19791095890411	-15.9450609375961\\
0.199931506849315	-15.933327860062\\
0.2	-15.9329235655716\\
0.201952054794521	-15.9212237930074\\
0.203972602739726	-15.9087595381636\\
0.205993150684932	-15.895945551435\\
0.208013698630137	-15.8827919568022\\
0.21	-15.8695397865513\\
0.210034246575342	-15.8693085595585\\
0.212054794520548	-15.8555048589148\\
0.214075342465753	-15.84139006001\\
0.216095890410959	-15.8269730853583\\
0.218116438356164	-15.8122625857635\\
0.22	-15.7982924472986\\
0.22013698630137	-15.7972669507309\\
0.222157534246575	-15.7819943184013\\
0.224178082191781	-15.766452585034\\
0.226198630136986	-15.7506494140625\\
0.228219178082192	-15.7345922447442\\
0.23	-15.7202353371556\\
0.230239726027397	-15.7182883004262\\
0.232260273972603	-15.7017445964469\\
0.234280821917808	-15.6849679476911\\
0.236301369863014	-15.6679649758174\\
0.238321917808219	-15.6507421161732\\
0.24	-15.6362757345346\\
0.240342465753425	-15.6333056244142\\
0.24236301369863	-15.6156615828422\\
0.244383561643836	-15.5978159064748\\
0.246404109589041	-15.5797743488615\\
0.248424657534247	-15.5615425076574\\
0.25	-15.547199346141\\
0.250445205479452	-15.5431258299657\\
0.252465753424658	-15.5245296174625\\
0.254486301369863	-15.5057590313106\\
0.256506849315068	-15.4868190968765\\
0.258527397260274	-15.4677147082556\\
0.26	-15.4536903026687\\
0.260547945205479	-15.4484506326184\\
0.262568493150685	-15.4290315143829\\
0.26458904109589	-15.4094618792235\\
0.266609589041096	-15.3897461379221\\
0.268630136986301	-15.3698885900699\\
0.27	-15.3563473186432\\
0.270650684931507	-15.3498934276265\\
0.272671232876712	-15.3297647383408\\
0.274691780821918	-15.3095065090436\\
0.276712328767123	-15.2891226288135\\
0.278732876712329	-15.2686168920244\\
0.28	-15.2556968491731\\
0.280753424657534	-15.247993001279\\
0.28277397260274	-15.227254570233\\
0.284794520547945	-15.2064051263148\\
0.286815068493151	-15.1854481133456\\
0.288835616438356	-15.1643868940636\\
0.29	-15.1522038759072\\
0.290856164383562	-15.1432247525568\\
0.292876712328767	-15.1219648966081\\
0.294897260273973	-15.1006104599561\\
0.296917808219178	-15.0791645044755\\
0.298938356164384	-15.0576300222806\\
0.3	-15.0462808067417\\
0.300958904109589	-15.0360099377543\\
0.302979452054795	-15.0143071095063\\
0.305	-14.9925243322631\\
0.307020547945205	-14.970664338693\\
0.309041095890411	-14.9487298011679\\
0.31	-14.9382948638098\\
0.311061643835616	-14.9267233334651\\
0.313082191780822	-14.9046474924112\\
0.315102739726027	-14.882504779471\\
0.317123287671233	-14.8602976422816\\
0.319143835616438	-14.8380284761372\\
0.32	-14.8285742516019\\
0.321164383561644	-14.815699625423\\
0.323184931506849	-14.7933133850024\\
0.325205479452055	-14.7708720015589\\
0.32722602739726	-14.7483776748935\\
0.329246575342466	-14.7258325591805\\
0.33	-14.7174133343681\\
0.331267123287671	-14.7032387641827\\
0.333287671232877	-14.680598356427\\
0.335308219178082	-14.6579133603433\\
0.337328767123288	-14.6351857593667\\
0.339349315068493	-14.6124174970048\\
0.34	-14.6050770040138\\
0.341369863013699	-14.5896104778718\\
0.343390410958904	-14.5667665686899\\
0.34541095890411	-14.54388759926\\
0.347431506849315	-14.5209753634018\\
0.349452054794521	-14.4980316198657\\
0.35	-14.4918043827612\\
0.351472602739726	-14.4750580932159\\
0.353493150684931	-14.4520564746875\\
0.355513698630137	-14.4290284230171\\
0.357534246575342	-14.4059755652482\\
0.359554794520548	-14.3828994975135\\
0.36	-14.3778119761689\\
0.361575342465753	-14.3598017857922\\
0.363595890410959	-14.3366839666466\\
0.365616438356164	-14.3135475479355\\
0.36763698630137	-14.2903940095076\\
0.369657534246575	-14.2672248038741\\
0.37	-14.2632963696837\\
0.371678082191781	-14.244041356862\\
0.373698630136986	-14.2208450682486\\
0.375719178082192	-14.1976373123773\\
0.377739726027397	-14.1744194387561\\
0.379760273972603	-14.1511927726391\\
0.38	-14.1484365442558\\
0.381780821917808	-14.127958615591\\
0.383801369863014	-14.1047182460367\\
0.385821917808219	-14.081472919794\\
0.387842465753425	-14.0582238705928\\
0.38986301369863	-14.0349723105789\\
0.39	-14.0333958725826\\
0.391883561643836	-14.0117194308043\\
0.393904109589041	-13.9884664017041\\
0.395924657534247	-13.9652143735598\\
0.397945205479452	-13.9419644769503\\
0.399965753424658	-13.9187178231911\\
0.4	-13.9183238464145\\
0.401986301369863	-13.8954755047603\\
0.404006849315068	-13.8722385957147\\
0.406027397260274	-13.8490081520938\\
0.408047945205479	-13.825785212313\\
0.41	-13.8033575764436\\
0.410068493150685	-13.802570797547\\
0.41208904109589	-13.7793659121026\\
0.414109589041096	-13.7561715437816\\
0.416130136986301	-13.7329886642345\\
0.418150684931507	-13.709818229305\\
0.42	-13.6886230991203\\
0.420171232876712	-13.6866611793654\\
0.422191780821918	-13.663518439643\\
0.424212328767123	-13.6403909205391\\
0.426232876712329	-13.6172795179388\\
0.428253424657534	-13.5941851135134\\
0.43	-13.5742365189257\\
0.43027397260274	-13.5711085750153\\
0.432294520547945	-13.5480507565647\\
0.434315068493151	-13.5250124989302\\
0.436335616438356	-13.5019946298014\\
0.438356164383562	-13.4789979640553\\
0.44	-13.460305009905\\
0.440376712328767	-13.4560233040155\\
0.442397260273973	-13.4330714397059\\
0.444417808219178	-13.4101431490974\\
0.446438356164384	-13.3872391983486\\
0.448458904109589	-13.364360342041\\
0.45	-13.3469276964021\\
0.450479452054795	-13.3415073234086\\
0.4525	-13.318680874561\\
0.454520547945205	-13.2958817167023\\
0.456541095890411	-13.2731105603443\\
0.458561643835616	-13.250368105514\\
0.46	-13.2341964297633\\
0.460582191780822	-13.2276550419573\\
0.462602739726027	-13.204972049337\\
0.464623287671233	-13.1823197974266\\
0.466643835616438	-13.159698946299\\
0.468664383561644	-13.1371101465118\\
0.47	-13.1221964751641\\
0.470684931506849	-13.1145540392871\\
0.472705479452055	-13.0920312566879\\
0.47472602739726	-13.0695424217902\\
0.476746575342466	-13.0470881488511\\
0.478767123287671	-13.0246690434734\\
0.48	-13.0110071205487\\
0.480787671232877	-13.0022857027657\\
0.482808219178082	-12.9799387154994\\
0.484828767123288	-12.9576286622626\\
0.486849315068493	-12.9353561156091\\
0.488869863013699	-12.9131216402058\\
0.49	-12.9007022178738\\
0.490890410958904	-12.8909257929754\\
0.49291095890411	-12.8687691232366\\
0.494931506849315	-12.8466521728414\\
0.496952054794521	-12.8245754763084\\
0.498972602739726	-12.8025395609544\\
0.5	-12.7913506653462\\
0.500993150684932	-12.7805449470223\\
0.503013698630137	-12.7585921478062\\
0.505034246575342	-12.7366816697745\\
0.507054794520548	-12.7148140126893\\
0.509075342465753	-12.6929896697241\\
0.51	-12.6830168380882\\
0.511095890410959	-12.6712091275787\\
0.513116438356164	-12.6494728665911\\
0.51513698630137	-12.6277813608481\\
0.517157534246575	-12.6061350782926\\
0.519178082191781	-12.5845344808289\\
0.52	-12.5757609736128\\
0.521198630136986	-12.5629800244261\\
0.523219178082192	-12.5414721592188\\
0.525239726027397	-12.5200113296062\\
0.527260273972603	-12.4985979743488\\
0.529280821917808	-12.4772325266631\\
0.53	-12.469639517594\\
0.531301369863014	-12.4559154143148\\
0.533321917808219	-12.4346470597091\\
0.535342465753425	-12.4134278799804\\
0.53736301369863	-12.3922582870792\\
0.539383561643836	-12.3711386878573\\
0.54	-12.3647054346752\\
0.541404109589041	-12.3500694841522\\
0.543424657534247	-12.3290510728685\\
0.545445205479452	-12.3080838460586\\
0.547465753424658	-12.2871681910013\\
0.549486301369863	-12.2663044902791\\
0.55	-12.2610084884105\\
0.551506849315069	-12.2454931218538\\
0.553527397260274	-12.2247344591405\\
0.555547945205479	-12.2040288710805\\
0.557568493150685	-12.1833767222124\\
0.55958904109589	-12.162778372742\\
0.56	-12.1585954938994\\
0.561609589041096	-12.1422341786107\\
0.563630136986301	-12.1217444915627\\
0.565650684931507	-12.1013096592109\\
0.567671232876712	-12.0809300251011\\
0.569691780821918	-12.0606059287759\\
0.57	-12.0575105462075\\
0.571712328767123	-12.0403377058363\\
0.573732876712329	-12.0201256880027\\
0.575753424657534	-11.9999702031748\\
0.57777397260274	-11.9798715754902\\
0.579794520547945	-11.9598301253817\\
0.58	-11.9577952272707\\
0.581815068493151	-11.9398461696337\\
0.583835616438356	-11.9199200214377\\
0.585856164383562	-11.9000519904467\\
0.587876712328767	-11.8802423828283\\
0.589897260273973	-11.8604915013168\\
0.59	-11.8594887936416\\
0.591917808219178	-11.8407996452649\\
0.593938356164384	-11.8211671106941\\
0.595958904109589	-11.8015941903437\\
0.597979452054795	-11.7820811737197\\
0.6	-11.7626283471423\\
0.6	-11.7626283471423\\
};
\addlegendentry{\small LASSO, Analysis}

\addplot [color=mycolor1, line width=1.0pt]
  table[row sep=crcr]{%
-0.0195	-12.0147317862576\\
-0.017277397260274	-12.2511139814688\\
-0.0150547945205479	-12.4803352666333\\
-0.0128321917808219	-12.7025635873058\\
-0.0106095890410959	-12.9179633048612\\
-0.00838698630136986	-13.1266952650177\\
-0.00616438356164384	-13.3289168652102\\
-0.00394178082191781	-13.5247821208274\\
-0.00171917808219178	-13.7144417303271\\
0.000503424657534245	-13.8980431392428\\
0.00272602739726027	-14.0757306030975\\
0.0049486301369863	-14.2476452492365\\
0.00717123287671232	-14.4139251375944\\
0.00939383561643835	-14.5747053204104\\
0.01	-14.6176174053626\\
0.0116164383561644	-14.7301179009053\\
0.0138390410958904	-14.8802920909337\\
0.0160616438356164	-15.0253542676268\\
0.0182842465753425	-15.1654280290381\\
0.02	-15.2702236497621\\
0.0205068493150685	-15.3006342488072\\
0.0227294520547945	-15.4310911298547\\
0.0249520547945205	-15.5569142571231\\
0.0271746575342466	-15.6782166493771\\
0.0293972602739726	-15.7951088100778\\
0.03	-15.8260623685616\\
0.0316198630136986	-15.9076987773453\\
0.0338424657534246	-16.0160921730225\\
0.0360650684931507	-16.120392250857\\
0.0382876712328767	-16.2206999438129\\
0.04	-16.2953180924484\\
0.0405102739726027	-16.3171139105288\\
0.0427328767123288	-16.4097305809362\\
0.0449554794520548	-16.4986442010526\\
0.0471780821917808	-16.5839468769652\\
0.0494006849315068	-16.6657286180188\\
0.05	-16.6871894400963\\
0.0516232876712329	-16.7440773792246\\
0.0538458904109589	-16.8190791029043\\
0.0560684931506849	-16.8908177595844\\
0.0582910958904109	-16.9593753881584\\
0.06	-17.0099748487373\\
0.060513698630137	-17.0248321353289\\
0.062736301369863	-17.0872662943499\\
0.064958904109589	-17.1467543430806\\
0.0671815068493151	-17.2033709813709\\
0.0694041095890411	-17.2571891677908\\
0.07	-17.2711515619337\\
0.0716267123287671	-17.3082801557235\\
0.0738493150684931	-17.3567135288355\\
0.0760719178082192	-17.4025572359408\\
0.0782945205479452	-17.4458776252773\\
0.08	-17.4774484425853\\
0.0805171232876712	-17.4867394782081\\
0.0827397260273972	-17.525206042367\\
0.0849623287671233	-17.5613390642629\\
0.0871849315068493	-17.5951988213589\\
0.0894075342465753	-17.6268441536427\\
0.09	-17.6349131983587\\
0.0916301369863013	-17.6563324947034\\
0.0938527397260274	-17.6837199023302\\
0.0960753424657534	-17.709061088648\\
0.0982979452054794	-17.7324094498056\\
0.1	-17.7489746191441\\
0.100520547945205	-17.7538170952283\\
0.102743150684931	-17.7733348764529\\
0.104965753424658	-17.7910124155541\\
0.107188356164384	-17.8068981331787\\
0.10941095890411	-17.8210392761965\\
0.11	-17.8245004233861\\
0.111633561643836	-17.8334819449811\\
0.113856164383562	-17.8442711203279\\
0.116078767123288	-17.8534506900147\\
0.118301369863014	-17.8610634750054\\
0.12	-17.8658512819386\\
0.12052397260274	-17.867151255284\\
0.122746575342466	-17.8717547952841\\
0.124969178082192	-17.8749138688386\\
0.127191780821918	-17.8766672834962\\
0.129414383561644	-17.8770529039162\\
0.13	-17.8769314958175\\
0.13163698630137	-17.8761076738104\\
0.133859589041096	-17.8738676355045\\
0.136082191780822	-17.8703679455401\\
0.138304794520548	-17.8656428837297\\
0.14	-17.8612360126561\\
0.140527397260274	-17.8597258515785\\
0.14275	-17.8526493538657\\
0.144972602739726	-17.8444449543574\\
0.147195205479452	-17.835143193142\\
0.149417808219178	-17.8247734492098\\
0.15	-17.821884195765\\
0.151640410958904	-17.8133637282795\\
0.15386301369863	-17.800940353649\\
0.156085616438356	-17.787527538745\\
0.158308219178082	-17.7731468263496\\
0.16	-17.7615632457179\\
0.160530821917808	-17.757816393801\\
0.162753424657534	-17.7415502481626\\
0.16497602739726	-17.7243573717229\\
0.167198630136986	-17.7062409252997\\
0.169421232876712	-17.6871976703306\\
0.17	-17.6820856369804\\
0.171643835616438	-17.6672178220066\\
0.173866438356164	-17.6462855817079\\
0.17608904109589	-17.6243806014476\\
0.178311643835616	-17.6014805892082\\
0.18	-17.5834072973619\\
0.180534246575342	-17.5775651591303\\
0.182756849315068	-17.5526208608159\\
0.184979452054794	-17.5266470980131\\
0.187202054794521	-17.4996623962456\\
0.189424657534247	-17.4717102455276\\
0.19	-17.4643251617815\\
0.191647260273973	-17.4428635829254\\
0.193869863013699	-17.4132269453945\\
0.196092465753425	-17.3829354574822\\
0.198315068493151	-17.3521501349533\\
0.2	-17.3285920312936\\
0.200537671232877	-17.3210494602892\\
0.202760273972603	-17.2898177552576\\
0.204982876712329	-17.2586314429652\\
0.207205479452055	-17.2276447458142\\
0.209428082191781	-17.1969766049074\\
0.21	-17.1891467262517\\
0.211650684931507	-17.1667005635185\\
0.213873287671233	-17.1368390186052\\
0.216095890410959	-17.1073626571489\\
0.218318493150685	-17.0781951588924\\
0.22	-17.0562652016042\\
0.220541095890411	-17.0492224969349\\
0.222763698630137	-17.0203055390366\\
0.224986301369863	-16.9912942549987\\
0.227208904109589	-16.9620417287356\\
0.229431506849315	-16.9324163568079\\
0.23	-16.9247653740846\\
0.231654109589041	-16.9023110311822\\
0.233876712328767	-16.8716486564871\\
0.236099315068493	-16.84038392976\\
0.238321917808219	-16.8085018122859\\
0.24	-16.7840275058805\\
0.240544520547945	-16.7760134766989\\
0.242767123287671	-16.7429506854771\\
0.244989726027397	-16.7093595557401\\
0.247212328767123	-16.6752945265469\\
0.249434931506849	-16.6408131220099\\
0.25	-16.631987117023\\
0.251657534246575	-16.6059718523538\\
0.253880136986301	-16.5708233624554\\
0.256102739726027	-16.535414754726\\
0.258325342465753	-16.4997868943854\\
0.26	-16.4728188839341\\
0.260547945205479	-16.463974449201\\
0.262770547945205	-16.428006411289\\
0.264993150684931	-16.391906879294\\
0.267215753424657	-16.3556959285147\\
0.269438356164384	-16.3193904504611\\
0.27	-16.310202985364\\
0.27166095890411	-16.2830048922948\\
0.273883561643836	-16.246551865434\\
0.276106164383562	-16.2100426197372\\
0.278328767123288	-16.1734873960242\\
0.28	-16.1459760404062\\
0.280551369863014	-16.1368956774843\\
0.28277397260274	-16.1002763623585\\
0.284996575342466	-16.0636378785574\\
0.287219178082192	-16.026988257434\\
0.289441780821918	-15.9903351800491\\
0.29	-15.9811298849521\\
0.291664383561644	-15.9536860056655\\
0.29388698630137	-15.9170477892354\\
0.296109589041096	-15.880427292379\\
0.298332191780822	-15.8438309907278\\
0.3	-15.8163892974712\\
0.300554794520548	-15.8072650794014\\
0.302777397260274	-15.7707354776716\\
0.305	-15.7342478334172\\
0.307222602739726	-15.6978075277049\\
0.309445205479452	-15.6614196796733\\
0.31	-15.6523454901404\\
0.311667808219178	-15.6250891518093\\
0.313890410958904	-15.5888205556538\\
0.31611301369863	-15.5526182579483\\
0.318335616438356	-15.516486387217\\
0.32	-15.4894777018796\\
0.320558219178082	-15.4804288407698\\
0.322780821917808	-15.4444492921068\\
0.325003424657534	-15.4085511986996\\
0.32722602739726	-15.3727378101223\\
0.329448630136986	-15.337012176501\\
0.33	-15.3281634979494\\
0.331671232876712	-15.3013771572517\\
0.333893835616438	-15.2658354300701\\
0.336116438356164	-15.230389500138\\
0.33833904109589	-15.1950417095102\\
0.34	-15.1686915677775\\
0.340561643835616	-15.1597942466413\\
0.342784246575342	-15.124649156014\\
0.345006849315068	-15.0896083478297\\
0.347229452054794	-15.054673607719\\
0.34945205479452	-15.0198466064337\\
0.35	-15.0112773338534\\
0.351674657534246	-14.9851289094782\\
0.353897260273973	-14.9505219866423\\
0.356119863013699	-14.9160272213947\\
0.358342465753425	-14.8816459201\\
0.36	-14.8560802284538\\
0.360565068493151	-14.8473793210216\\
0.362787671232877	-14.8132286030756\\
0.365010273972603	-14.779194894301\\
0.367232876712329	-14.7452792800146\\
0.369455479452055	-14.71148281062\\
0.37	-14.703221223627\\
0.371678082191781	-14.677806509043\\
0.373900684931507	-14.6442513777668\\
0.376123287671233	-14.6108184054453\\
0.378345890410959	-14.5775085730716\\
0.38	-14.5527991714404\\
0.380568493150685	-14.5443228596843\\
0.382791095890411	-14.5112622475965\\
0.385013698630137	-14.4783277271334\\
0.387236301369863	-14.4455203008689\\
0.389458904109589	-14.4128409873532\\
0.39	-14.4049046688039\\
0.391681506849315	-14.3802908243271\\
0.393904109589041	-14.3478708714196\\
0.396126712328767	-14.3155822123308\\
0.398349315068493	-14.2834259565015\\
0.4	-14.2596304678389\\
0.400571917808219	-14.2514032402752\\
0.402794520547945	-14.2195152275612\\
0.405017123287671	-14.1877631100076\\
0.407239726027397	-14.1561481066964\\
0.409462328767123	-14.1246714633755\\
0.41	-14.1170778547266\\
0.411684931506849	-14.093334451243\\
0.413907534246575	-14.0621383653011\\
0.416130136986301	-14.0310845223002\\
0.418352739726027	-14.0001742582919\\
0.42	-13.9773588546433\\
0.420575342465753	-13.969408925814\\
0.422797945205479	-13.9387898907306\\
0.425020547945205	-13.9083185287496\\
0.427243150684931	-13.8779962216444\\
0.429465753424657	-13.8478243532026\\
0.43	-13.8405945275879\\
0.431688356164384	-13.8178043049291\\
0.43391095890411	-13.7879374515287\\
0.436133561643836	-13.7582251561939\\
0.438356164383562	-13.7286687657241\\
0.44	-13.706909951247\\
0.440578767123288	-13.6992696055028\\
0.442801369863014	-13.6700289743569\\
0.44502397260274	-13.6409481393237\\
0.447246575342466	-13.6120283303511\\
0.449469178082192	-13.5832707349538\\
0.45	-13.5764267119577\\
0.451691780821918	-13.5546764928493\\
0.453914383561644	-13.5262466905958\\
0.45613698630137	-13.4979823562542\\
0.458359589041096	-13.4698844540947\\
0.46	-13.4492538326735\\
0.460582191780822	-13.441953879367\\
0.462804794520548	-13.4141914531539\\
0.465027397260274	-13.3865979173245\\
0.46725	-13.359173929605\\
0.469472602739726	-13.3319200587815\\
0.47	-13.3254780652559\\
0.471695205479452	-13.3048367800488\\
0.473917808219178	-13.2779244705199\\
0.476140410958904	-13.2511834049062\\
0.47836301369863	-13.2246137513804\\
0.48	-13.2051543847411\\
0.480585616438356	-13.1982155676326\\
0.482808219178082	-13.1719887971265\\
0.485030821917808	-13.1459332655656\\
0.487253424657534	-13.1200486775742\\
0.48947602739726	-13.0943346136001\\
0.49	-13.0882973745682\\
0.491698630136986	-13.0687905270442\\
0.493921232876712	-13.0434157416195\\
0.496143835616438	-13.0182094489445\\
0.498366438356164	-12.9931707063715\\
0.5	-12.9748740147505\\
0.50058904109589	-12.9682984350534\\
0.502811643835616	-12.9435914182481\\
0.505034246575342	-12.9190482998625\\
0.507256849315069	-12.8946675832344\\
0.509479452054795	-12.8704476301525\\
0.51	-12.8647982067809\\
0.511702054794521	-12.8463866601129\\
0.513924657534247	-12.8224827498103\\
0.516147260273973	-12.7987338328618\\
0.518369863013699	-12.775137699761\\
0.52	-12.7579272099269\\
0.520592465753425	-12.7516919980595\\
0.522815068493151	-12.7283942327723\\
0.525037671232877	-12.705241767005\\
0.527260273972603	-12.6822318227976\\
0.529482876712329	-12.6593614821827\\
0.53	-12.6540600354471\\
0.531705479452055	-12.6366276884532\\
0.533928082191781	-12.6140272476361\\
0.536150684931507	-12.591556830168\\
0.538373287671233	-12.5692129727682\\
0.54	-12.5529377521141\\
0.540595890410959	-12.5469920805054\\
0.542818493150685	-12.524890429053\\
0.545041095890411	-12.5029041671293\\
0.547263698630137	-12.4810293191181\\
0.549486301369863	-12.4592617878645\\
0.55	-12.4542455958488\\
0.551708904109589	-12.4375973576432\\
0.553931506849315	-12.4160316972922\\
0.556154109589041	-12.3945603635097\\
0.558376712328767	-12.3731788043093\\
0.56	-12.3576167416822\\
0.560599315068493	-12.3518823626271\\
0.562821917808219	-12.3306662800786\\
0.565044520547945	-12.3095257008594\\
0.567267123287671	-12.288455675785\\
0.569489726027397	-12.2674511664657\\
0.57	-12.262637579257\\
0.571712328767123	-12.2465070496117\\
0.573934931506849	-12.2256181214624\\
0.576157534246575	-12.2047791023375\\
0.578380136986301	-12.1839846413028\\
0.58	-12.1688543253865\\
0.580602739726027	-12.1632293209462\\
0.582825342465753	-12.1425076622598\\
0.585047945205479	-12.1218141296222\\
0.587270547945205	-12.1011431358755\\
0.589493150684931	-12.0804890474921\\
0.59	-12.0757808020364\\
0.591715753424657	-12.0598461898263\\
0.593938356164384	-12.0392088524439\\
0.596160958904109	-12.0185712945257\\
0.598383561643836	-11.9979277503378\\
0.6	-11.9829072009802\\
0.600606164383562	-11.9772724347637\\
0.602828767123288	-11.9565995488917\\
0.605051369863014	-11.9359032856519\\
0.60727397260274	-11.9151778354966\\
0.609496575342466	-11.8944173921169\\
0.611719178082192	-11.8736161581907\\
0.613941780821918	-11.8527683511538\\
0.616164383561644	-11.8318682089887\\
0.61838698630137	-11.810909996024\\
0.620609589041096	-11.7898880087375\\
0.622832191780822	-11.7687965815555\\
0.625054794520548	-11.7476300926431\\
0.627277397260274	-11.7263829696761\\
0.6295	-11.7050496955891\\
};
\addlegendentry{\small box-LASSO, Analysis}

\addplot [color=blue, draw=none, mark size=3.5pt, mark=square, mark options={solid, blue}]
  table[row sep=crcr]{%
0.01	-14.6960802536532\\
0.11	-17.7572446584973\\
0.21	-17.0912445901946\\
0.31	-15.638794757375\\
0.41	-14.0354547616817\\
0.51	-12.8661818774364\\
};
\addlegendentry{\small box-LASSO, Simulations}

\addplot [color=black, draw=none, mark size=5.0pt, mark=x, mark options={solid, black}]
  table[row sep=crcr]{%
0.129414383561644	-17.8770529039162\\
};

\addplot [color=black, dashed]
  table[row sep=crcr]{%
0.129414383561644	-17.8770529039162\\
0.129414383561644	-20\\
};

\addplot [color=black, dashed, forget plot]
  table[row sep=crcr]{%
0	-17.88\\
0.129414383561644	-17.88\\
};

\end{axis}
\end{tikzpicture}%

%% file: Figs/PEvsLambda.tex
%
%
\definecolor{mycolor1}{rgb}{0.00000,0.44700,0.74100}%
\begin{tikzpicture}
\begin{axis}[%
width=2.8in,
height=2.1in,
at={(1.989in,1.234in)},
scale only axis,
xmin=0.01,
xmax=0.6,
xtick={0.3,0.6},
xticklabels={{$0.3$},{$0.6$}},
xlabel style={font=\color{white!15!black}},
xlabel={$\lambda$},
ymode=log,
ymin=0.0065,
ymax=0.115,
yminorticks=true,
ylabel style={font=\color{white!15!black}},
ylabel={$\bar{P}_{\rm E}$},
axis background/.style={fill=white},
legend style={at={(0.36,0.03)}, anchor=south west, legend cell align=left, align=left, draw=white!15!black}
]

\addplot [color=red, line width=1.0pt]
  table[row sep=crcr]{%
0.01	0.103229910709193\\
0.01	0.103229910709193\\
0.0120205479452055	0.0869430640642919\\
0.014041095890411	0.0749688024429553\\
0.0160616438356164	0.065814706777015\\
0.0180821917808219	0.0586060739638962\\
0.02	0.0530641336846034\\
0.0201027397260274	0.0527962876500843\\
0.0221232876712329	0.048026176435746\\
0.0241438356164384	0.044050006398874\\
0.0261643835616438	0.0406939712702446\\
0.0281849315068493	0.0378316670968865\\
0.03	0.0356033677868262\\
0.0302054794520548	0.0353689612426268\\
0.0322260273972603	0.0332343094505103\\
0.0342465753424658	0.0313723619637099\\
0.0362671232876712	0.0297396251255655\\
0.0382876712328767	0.0283014467448291\\
0.04	0.0272140343381995\\
0.0403082191780822	0.0270298767941341\\
0.0423287671232877	0.0259021206056866\\
0.0443493150684931	0.0248994015648123\\
0.0463698630136986	0.0240061121681445\\
0.0483904109589041	0.0232091716040876\\
0.05	0.0226358402133447\\
0.0504109589041096	0.0224975335252413\\
0.0524315068493151	0.0218618045833144\\
0.0544520547945205	0.0212939456994897\\
0.056472602739726	0.020787035865451\\
0.0584931506849315	0.0203350837192756\\
0.06	0.0200306781648816\\
0.060513698630137	0.0199328759900781\\
0.0625342465753425	0.0195758546607622\\
0.0645547945205479	0.0192600166944793\\
0.0665753424657534	0.0189818316329487\\
0.0685958904109589	0.0187381734576091\\
0.07	0.0185877026218568\\
0.0706164383561644	0.0185262639141398\\
0.0726369863013699	0.0183436251117248\\
0.0746575342465753	0.0181880396733105\\
0.0766780821917808	0.0180575170697696\\
0.0786986301369863	0.0179502650466169\\
0.08	0.0178927462439695\\
0.0807191780821918	0.0178646652666341\\
0.0827397260273973	0.0177992524600814\\
0.0847602739726027	0.0177526965069985\\
0.0867808219178082	0.0177237869815346\\
0.0888013698630137	0.017711419772444\\
0.09	0.0177114846602099\\
0.0908219178082192	0.017714585461485\\
0.0928424657534246	0.0177323591960207\\
0.0948630136986301	0.0177638918363815\\
0.0968835616438356	0.0178084021946234\\
0.0989041095890411	0.0178651702108531\\
0.1	0.017900848882008\\
0.100924657534247	0.0179335309375809\\
0.102945205479452	0.0180128692226214\\
0.104965753424658	0.0181026149976927\\
0.106986301369863	0.0182022390937061\\
0.109006849315068	0.0183112495153012\\
0.11	0.018368130871227\\
0.111027397260274	0.0184291881168753\\
0.113047945205479	0.0185556276305071\\
0.115068493150685	0.0186901690030596\\
0.11708904109589	0.0188324390055687\\
0.119109589041096	0.0189820880829721\\
0.12	0.019050292128305\\
0.121130136986301	0.019138788416448\\
0.123150684931507	0.0193022321742259\\
0.125171232876712	0.0194721299298186\\
0.127191780821918	0.0196482092292664\\
0.129212328767123	0.0198302132912628\\
0.13	0.0199027163016831\\
0.131232876712329	0.0200178998259967\\
0.133253424657534	0.0202110399602456\\
0.13527397260274	0.0204094172577284\\
0.137294520547945	0.0206128268250062\\
0.139315068493151	0.0208210744943363\\
0.14	0.0208927320450523\\
0.141335616438356	0.0210339760758552\\
0.143356164383562	0.0212513566723209\\
0.145376712328767	0.0214730500503861\\
0.147397260273973	0.0216988980630327\\
0.149417808219178	0.0219287501183695\\
0.15	0.0219957012568859\\
0.151438356164384	0.0221624626905054\\
0.153458904109589	0.0223998988686561\\
0.155479452054795	0.0226409279410382\\
0.1575	0.0228854250104557\\
0.159520547945205	0.0231332706387947\\
0.16	0.0231925604070871\\
0.161541095890411	0.0233843505179169\\
0.163561643835616	0.0236385551646903\\
0.165582191780822	0.023895779638113\\
0.167602739726027	0.0241559232766815\\
0.169623287671233	0.0244188894543304\\
0.17	0.0244682216787904\\
0.171643835616438	0.0246845853534267\\
0.173664383561644	0.0249529217534391\\
0.175684931506849	0.0252238128340343\\
0.177705479452055	0.0254971759914584\\
0.17972602739726	0.0257729316671682\\
0.18	0.025810502230544\\
0.181746575342466	0.0260510031877662\\
0.183767123287671	0.0263313166153765\\
0.185787671232877	0.0266138006076729\\
0.187808219178082	0.0268983862868379\\
0.189828767123288	0.0271850071167934\\
0.19	0.0272093884910268\\
0.191849315068493	0.0274735987880958\\
0.193869863013699	0.0277640991099437\\
0.195890410958904	0.0280564479087857\\
0.19791095890411	0.0283505869330619\\
0.199931506849315	0.0286464597636486\\
0.2	0.0286565190794357\\
0.201952054794521	0.0289440117296082\\
0.203972602739726	0.029243189828881\\
0.205993150684932	0.0295439426535811\\
0.208013698630137	0.0298462203195855\\
0.21	0.030144813986867\\
0.210034246575342	0.0301499744001308\\
0.212054794520548	0.03045515786315\\
0.214075342465753	0.0307617250121062\\
0.216095890410959	0.0310696314300943\\
0.218116438356164	0.0313788339270022\\
0.22	0.031668203796455\\
0.22013698630137	0.0316892904895341\\
0.222157534246575	0.0320009602339168\\
0.224178082191781	0.0323138033611203\\
0.226198630136986	0.0326277811144367\\
0.228219178082192	0.0329428557392721\\
0.23	0.0332214286966671\\
0.230239726027397	0.0332589904450163\\
0.232260273972603	0.0335761493688664\\
0.234280821917808	0.0338942975414844\\
0.236301369863014	0.0342134008543819\\
0.238321917808219	0.03453342602893\\
0.24	0.0347998870634005\\
0.240342465753425	0.0348543405868984\\
0.24236301369863	0.0351761128224359\\
0.244383561643836	0.0354987117754096\\
0.246404109589041	0.0358221072060245\\
0.248424657534247	0.0361462695706526\\
0.25	0.0363995198208514\\
0.250445205479452	0.0364711699988021\\
0.252465753424658	0.0367967802711646\\
0.254486301369863	0.0371230727986798\\
0.256506849315068	0.0374500206025633\\
0.258527397260274	0.0377775972952435\\
0.26	0.0380167210105026\\
0.260547945205479	0.0381057770621606\\
0.262568493150685	0.0384345346443796\\
0.26458904109589	0.038763845321976\\
0.266609589041096	0.0390936848981521\\
0.268630136986301	0.0394240296840469\\
0.27	0.0396482678179814\\
0.270650684931507	0.0397548564842021\\
0.272671232876712	0.0400861425826529\\
0.274691780821918	0.0404178657296085\\
0.276712328767123	0.0407500041286957\\
0.278732876712329	0.0410825364247347\\
0.28	0.0412912652260497\\
0.280753424657534	0.041415441692022\\
0.28277397260274	0.0417486994230943\\
0.284794520547945	0.0420822895179513\\
0.286815068493151	0.0424161922737126\\
0.288835616438356	0.0427503883746895\\
0.29	0.0429431017880515\\
0.290856164383562	0.0430848588828517\\
0.292876712328767	0.0434195852286686\\
0.294897260273973	0.0437545492023087\\
0.296917808219178	0.0440897329451795\\
0.298938356164384	0.044425118941793\\
0.3	0.0446014139463069\\
0.300958904109589	0.0447606900119402\\
0.302979452054795	0.0450964293031609\\
0.305	0.0454323202834962\\
0.307020547945205	0.045768346734509\\
0.309041095890411	0.0461044927445625\\
0.31	0.0462640569813109\\
0.311061643835616	0.046440742702343\\
0.313082191780822	0.0467770812906176\\
0.315102739726027	0.0471134934802156\\
0.317123287671233	0.0474499645242242\\
0.319143835616438	0.0477864799523891\\
0.32	0.0479290811538006\\
0.321164383561644	0.0481230255657105\\
0.323184931506849	0.0484595874312265\\
0.325205479452055	0.0487961518769763\\
0.32722602739726	0.0491327054871348\\
0.329246575342466	0.0494692350973115\\
0.33	0.0495947119486542\\
0.331267123287671	0.0498057277900069\\
0.333287671232877	0.0501421708902205\\
0.335308219178082	0.0504785519612021\\
0.337328767123288	0.0508148588003433\\
0.339349315068493	0.0511510794352008\\
0.34	0.051259333585055\\
0.341369863013699	0.0514872021196481\\
0.343390410958904	0.0518232153301488\\
0.34541095890411	0.0521591077621484\\
0.347431506849315	0.0524948683265776\\
0.349452054794521	0.0528304861464654\\
0.35	0.0529214751474037\\
0.351472602739726	0.0531659505536547\\
0.353493150684931	0.0535012510856192\\
0.355513698630137	0.0538363774823752\\
0.357534246575342	0.0541713196834866\\
0.359554794520548	0.0545060678251582\\
0.36	0.0545797988341963\\
0.361575342465753	0.0548406122374143\\
0.363595890410959	0.0551749434413598\\
0.365616438356164	0.05550905214652\\
0.36763698630137	0.0558429292482566\\
0.369657534246575	0.0561765658252568\\
0.37	0.0562330899302387\\
0.371678082191781	0.0565099531370929\\
0.373698630136986	0.0568430826218497\\
0.375719178082192	0.0571759458938178\\
0.377739726027397	0.057508534741249\\
0.379760273972603	0.0578408411241731\\
0.38	0.0578802481902076\\
0.381780821917808	0.0581728571722733\\
0.383801369863014	0.0585045751828174\\
0.385821917808219	0.0588359876186442\\
0.387842465753425	0.0591670871062018\\
0.38986301369863	0.0594978664336368\\
0.39	0.0595202803852131\\
0.391883561643836	0.0598283185489324\\
0.393904109589041	0.0601584365580939\\
0.395924657534247	0.0604882137233791\\
0.397945205479452	0.0608176434615735\\
0.399965753424658	0.0611467193423074\\
0.4	0.0611522938133925\\
0.401986301369863	0.0614754350864143\\
0.404006849315068	0.061803784564329\\
0.406027397260274	0.062131761794524\\
0.408047945205479	0.0624593609419826\\
0.41	0.0627754906141515\\
0.410068493150685	0.0627865763167087\\
0.41208904109589	0.0631134023722702\\
0.414109589041096	0.0634398337043768\\
0.416130136986301	0.0637658650494895\\
0.418150684931507	0.0640914912834621\\
0.42	0.0643891627560094\\
0.420171232876712	0.0644167074202127\\
0.422191780821918	0.0647415086104249\\
0.424212328767123	0.0650658901402767\\
0.426232876712329	0.0653898474301981\\
0.428253424657534	0.0657133760336545\\
0.43	0.0659926875920255\\
0.43027397260274	0.0660364716359566\\
0.432294520547945	0.0663591300530945\\
0.434315068493151	0.0666813472305973\\
0.436335616438356	0.067003119242415\\
0.438356164383562	0.0673244422898243\\
0.44	0.0675855238959086\\
0.440376712328767	0.0676453127003556\\
0.442397260273973	0.0679657269267423\\
0.444417808219178	0.0682856815458903\\
0.446438356164384	0.0686051732578681\\
0.448458904109589	0.0689241988849163\\
0.45	0.0691672083072235\\
0.450479452054795	0.0692427553704762\\
0.4525	0.0695608397782369\\
0.454520547945205	0.0698784492912002\\
0.456541095890411	0.0701955812107634\\
0.458561643835616	0.0705122329558182\\
0.46	0.0707373521264432\\
0.460582191780822	0.0708284020618668\\
0.462602739726027	0.0711440861801537\\
0.464623287671233	0.0714592830768129\\
0.466643835616438	0.0717739906320301\\
0.468664383561644	0.0720882068392203\\
0.47	0.0722956384105683\\
0.470684931506849	0.0724019298042188\\
0.472705479452055	0.0727151577444864\\
0.47472602739726	0.0730278889883281\\
0.476746575342466	0.073340121974125\\
0.478767123287671	0.073651855249579\\
0.48	0.0738418193281566\\
0.480787671232877	0.0739630874709697\\
0.482808219178082	0.0742738174024236\\
0.484828767123288	0.0745840439151949\\
0.486849315068493	0.0748937659869581\\
0.488869863013699	0.0752029827011115\\
0.49	0.0753757137392324\\
0.490890410958904	0.0755116932460914\\
0.49291095890411	0.0758198969146975\\
0.494931506849315	0.0761275931034282\\
0.496952054794521	0.0764347813118255\\
0.498972602739726	0.0767414611418312\\
0.5	0.0768972049709934\\
0.500993150684932	0.0770476322971508\\
0.503013698630137	0.077353294582628\\
0.505034246575342	0.0776584479036283\\
0.507054794520548	0.0779630922654301\\
0.509075342465753	0.0782672277726266\\
0.51	0.0784062387647157\\
0.511095890410959	0.078570854628534\\
0.513116438356164	0.0788739731346093\\
0.51513698630137	0.0791765836898759\\
0.517157534246575	0.0794786867903568\\
0.519178082191781	0.0797802830285151\\
0.52	0.0799028213729759\\
0.521198630136986	0.0800813730927031\\
0.523219178082192	0.0803819577666175\\
0.525239726027397	0.0806820379287625\\
0.527260273972603	0.0809816145519197\\
0.529280821917808	0.0812806887026244\\
0.53	0.0813870177893962\\
0.531301369863014	0.0815792615406489\\
0.533321917808219	0.081877334318493\\
0.535342465753425	0.082174908380879\\
0.53736301369863	0.0824719851642548\\
0.539383561643836	0.0827685661963024\\
0.54	0.0828589500956909\\
0.541404109589041	0.0830646530954514\\
0.543424657534247	0.0833602475704004\\
0.545445205479452	0.0836553514196416\\
0.547465753424658	0.0839499665309934\\
0.549486301369863	0.0842440948811361\\
0.55	0.0843187959129561\\
0.551506849315069	0.0845377385351551\\
0.553527397260274	0.0848308996460876\\
0.555547945205479	0.0851235804544753\\
0.557568493150685	0.085415783287922\\
0.55958904109589	0.085707510560656\\
0.56	0.0857667869459513\\
0.561609589041096	0.0859987647730968\\
0.563630136986301	0.0862895485114278\\
0.565650684931507	0.0865798644471718\\
0.567671232876712	0.086869715336773\\
0.569691780821918	0.0871591040211814\\
0.57	0.0872032076106537\\
0.571712328767123	0.0874480334254432\\
0.573732876712329	0.0877365065582945\\
0.575753424657534	0.0880245265117593\\
0.57777397260274	0.088312096460752\\
0.579794520547945	0.0885992196626832\\
0.58	0.0886283937366582\\
0.581815068493151	0.0888858994570703\\
0.583835616438356	0.089172139265151\\
0.585856164383562	0.0894579425895012\\
0.587876712328767	0.0897433130136567\\
0.589897260273973	0.0900282542017377\\
0.59	0.0900427313370952\\
0.591917808219178	0.090312769898078\\
0.593938356164384	0.0905968639268569\\
0.595958904109589	0.0908805401917348\\
0.597979452054795	0.0911638026754921\\
0.6	0.0914466554396715\\
0.6	0.0914466554396715\\
};
\addlegendentry{\small LASSO, Analysis}

\addplot [color=mycolor1, line width=1.0pt]
  table[row sep=crcr]{%
0.01	0.0299926203087703\\
0.01	0.0299926203087703\\
0.0120205479452055	0.0285759709625983\\
0.014041095890411	0.0272397577010912\\
0.0160616438356164	0.0259805816390721\\
0.0180821917808219	0.0247951522725424\\
0.02	0.0237353195779772\\
0.0201027397260274	0.0236802854698515\\
0.0221232876712329	0.0226329014610539\\
0.0241438356164384	0.0216500228265287\\
0.0261643835616438	0.0207287724859812\\
0.0281849315068493	0.0198663716887768\\
0.03	0.0191396378986395\\
0.0302054794520548	0.0190601380067261\\
0.0322260273972603	0.0183074833302873\\
0.0342465753424658	0.0176059118691809\\
0.0362671232876712	0.0169530181584321\\
0.0382876712328767	0.0163464850707574\\
0.04	0.0158671061472012\\
0.0403082191780822	0.0157840818362467\\
0.0423287671232877	0.0152636620703097\\
0.0443493150684931	0.014783161810714\\
0.0463698630136986	0.0143405975646946\\
0.0483904109589041	0.0139340643668836\\
0.05	0.0136347774446375\\
0.0504109589041096	0.0135617338490544\\
0.0524315068493151	0.0132218523223461\\
0.0544520547945205	0.0129127388729382\\
0.056472602739726	0.0126327834718218\\
0.0584931506849315	0.0123804450995111\\
0.06	0.0122093575809717\\
0.060513698630137	0.0121542498864585\\
0.0625342465753425	0.01195278926979\\
0.0645547945205479	0.0117747181672318\\
0.0665753424657534	0.0116187531687285\\
0.0685958904109589	0.0114836707465746\\
0.07	0.0114014790714554\\
0.0706164383561644	0.0113683054845743\\
0.0726369863013699	0.0112715483269513\\
0.0746575342465753	0.0111923448475081\\
0.0766780821917808	0.0111296935396668\\
0.0786986301369863	0.0110826441279774\\
0.08	0.0110601791648853\\
0.0807191780821918	0.011050295901517\\
0.0827397260273973	0.0110317960697849\\
0.0847602739726027	0.011026338141545\\
0.0867808219178082	0.011033160327083\\
0.0888013698630137	0.0110515439643054\\
0.09	0.0110676347906267\\
0.0908219178082192	0.0110808119691146\\
0.0928424657534246	0.011120327310487\\
0.0948630136986301	0.0111694915105585\\
0.0968835616438356	0.0112277431701454\\
0.0989041095890411	0.0112945565199746\\
0.1	0.011334199783201\\
0.100924657534247	0.0113694399979955\\
0.102945205479452	0.0114519348529531\\
0.104965753424658	0.0115416137746706\\
0.106986301369863	0.0116380795510659\\
0.109006849315068	0.0117409637522726\\
0.11	0.0117937818775165\\
0.111027397260274	0.011849925442077\\
0.113047945205479	0.0119646499167299\\
0.115068493150685	0.0120848474713938\\
0.11708904109589	0.0122102521943492\\
0.119109589041096	0.0123406207890419\\
0.12	0.0123995890393327\\
0.121130136986301	0.0124757314240797\\
0.123150684931507	0.0126153826113135\\
0.125171232876712	0.0127593921119211\\
0.127191780821918	0.01290759587068\\
0.129212328767123	0.0130598469783995\\
0.13	0.0131202667985957\\
0.131232876712329	0.013216014662395\\
0.133253424657534	0.0133759833051589\\
0.13527397260274	0.0135396514910702\\
0.137294520547945	0.0137069310811139\\
0.139315068493151	0.0138777463155326\\
0.14	0.013936440497732\\
0.141335616438356	0.0140520329442989\\
0.143356164383562	0.0142297373853038\\
0.145376712328767	0.0144108159101408\\
0.147397260273973	0.0145952338572941\\
0.149417808219178	0.0147829648725875\\
0.15	0.0148376688583689\\
0.151438356164384	0.0149739901767494\\
0.153458904109589	0.0151682978597981\\
0.155479452054795	0.0153658822020915\\
0.1575	0.0155667430218246\\
0.159520547945205	0.0157708850486926\\
0.16	0.0158198081082955\\
0.161541095890411	0.015978317323453\\
0.163561643835616	0.0161890526230895\\
0.165582191780822	0.0164031069113866\\
0.167602739726027	0.0166204988144742\\
0.169623287671233	0.0168412491210982\\
0.17	0.0168827791875588\\
0.171643835616438	0.0170653803072814\\
0.173664383561644	0.017292916085025\\
0.175684931506849	0.0175238809746365\\
0.177705479452055	0.0177582999004165\\
0.17972602739726	0.0179961978092444\\
0.18	0.0180287243528667\\
0.181746575342466	0.0182375993116856\\
0.183767123287671	0.0184825283452669\\
0.185787671232877	0.0187310078594081\\
0.187808219178082	0.0189830595216853\\
0.189828767123288	0.0192387034449563\\
0.19	0.0192605338971786\\
0.191849315068493	0.0194979579348812\\
0.193869863013699	0.0197608392574332\\
0.195890410958904	0.0200273614258999\\
0.19791095890411	0.0202975360069765\\
0.199931506849315	0.0205713719453738\\
0.2	0.0205807187615658\\
0.201952054794521	0.0208488754065848\\
0.203972602739726	0.0211300496372103\\
0.205993150684932	0.0214148948424504\\
0.208013698630137	0.0217034080801993\\
0.21	0.0219906005919913\\
0.210034246575342	0.0219955831712861\\
0.212054794520548	0.022291410625313\\
0.214075342465753	0.0225908775816314\\
0.216095890410959	0.0228939677649209\\
0.218116438356164	0.0232006614548333\\
0.22	0.0234897873277057\\
0.22013698630137	0.0235109354692339\\
0.222157534246575	0.0238247631604937\\
0.224178082191781	0.0241421144242869\\
0.226198630136986	0.0244629557204312\\
0.228219178082192	0.0247872501052045\\
0.23	0.0250758997294495\\
0.230239726027397	0.025114957274603\\
0.232260273972603	0.0254460336180837\\
0.234280821917808	0.025780432282198\\
0.236301369863014	0.0261181032436279\\
0.238321917808219	0.0264589933911334\\
0.24	0.0267445123844332\\
0.240342465753425	0.0268030466158017\\
0.24236301369863	0.0271502039092183\\
0.244383561643836	0.0275004034689188\\
0.246404109589041	0.0278535808107207\\
0.248424657534247	0.028209668887313\\
0.25	0.0284892716650564\\
0.250445205479452	0.028568598212745\\
0.252465753424658	0.0289302969921829\\
0.254486301369863	0.0292946912565458\\
0.256506849315068	0.0296617050014497\\
0.258527397260274	0.0300312603300724\\
0.26	0.0303021528649745\\
0.260547945205479	0.0304032775993645\\
0.262568493150685	0.0307776755692143\\
0.26458904109589	0.0311543715541014\\
0.266609589041096	0.0315332815767123\\
0.268630136986301	0.0319143205231605\\
0.27	0.0321738192701221\\
0.270650684931507	0.0322974022993439\\
0.272671232876712	0.0326824399879825\\
0.274691780821918	0.033069346005929\\
0.276712328767123	0.0334580322613794\\
0.278732876712329	0.0338484103104955\\
0.28	0.0340940472139942\\
0.280753424657534	0.0342403915131802\\
0.28277397260274	0.0346338871874476\\
0.284794520547945	0.0350288087621639\\
0.286815068493151	0.0354250679277307\\
0.288835616438356	0.0358225767843615\\
0.29	0.036052183175134\\
0.290856164383562	0.0362212479876284\\
0.292876712328767	0.0366209948908951\\
0.294897260273973	0.0370217316844526\\
0.296917808219178	0.037423373530831\\
0.298938356164384	0.0378258366962492\\
0.3	0.0380376016491082\\
0.300958904109589	0.0382290386777101\\
0.302979452054795	0.038632898325588\\
0.305	0.0390373359613986\\
0.307020547945205	0.0394422734905582\\
0.309041095890411	0.0398476345098301\\
0.31	0.0400401358055096\\
0.311061643835616	0.0402533444093263\\
0.313082191780822	0.0406593304687741\\
0.315102739726027	0.0410655219479116\\
0.317123287671233	0.0414718501708386\\
0.319143835616438	0.041878248604125\\
0.32	0.0420504567544892\\
0.321164383561644	0.0422846529285275\\
0.323184931506849	0.042691001104253\\
0.325205479452055	0.043097233429532\\
0.32722602739726	0.0435032925924787\\
0.329246575342466	0.0439091237161345\\
0.33	0.0440603815180236\\
0.331267123287671	0.0443146743966088\\
0.333287671232877	0.0447198947342133\\
0.335308219178082	0.0451247373575919\\
0.337328767123288	0.0455291574408588\\
0.339349315068493	0.0459331127135418\\
0.34	0.0460630944460663\\
0.341369863013699	0.0463365634635827\\
0.343390410958904	0.0467394725331628\\
0.34541095890411	0.0471418053075673\\
0.347431506849315	0.0475435296969564\\
0.349452054794521	0.0479446161112864\\
0.35	0.0480532717469986\\
0.351472602739726	0.0483450374282817\\
0.353493150684931	0.0487447689545983\\
0.355513698630137	0.0491437883803232\\
0.357534246575342	0.0495420757268764\\
0.359554794520548	0.0499396132884114\\
0.36	0.0500271039221844\\
0.361575342465753	0.0503363855669008\\
0.363595890410959	0.0507323792010823\\
0.365616438356164	0.0511275828893436\\
0.36763698630137	0.0515219873068112\\
0.369657534246575	0.0519155850167994\\
0.37	0.0519822161105435\\
0.371678082191781	0.0523083703768362\\
0.373698630136986	0.0527003394394951\\
0.375719178082192	0.0530914898482405\\
0.377739726027397	0.0534818207285536\\
0.379760273972603	0.0538713325746489\\
0.38	0.053917491562866\\
0.381780821917808	0.0542600271319287\\
0.383801369863014	0.0546479072755925\\
0.385821917808219	0.0550349768855943\\
0.387842465753425	0.0554212407183827\\
0.38986301369863	0.0558067042755929\\
0.39	0.0558328085803842\\
0.391883561643836	0.0561913736701237\\
0.393904109589041	0.0565752554899364\\
0.395924657534247	0.0569583566599007\\
0.397945205479452	0.0573406843020373\\
0.399965753424658	0.0577222455946216\\
0.4	0.0577287061719513\\
0.401986301369863	0.058103047630416\\
0.404006849315068	0.0584830972745002\\
0.406027397260274	0.0588624010221187\\
0.408047945205479	0.059240964856857\\
0.41	0.0596059983158014\\
0.410068493150685	0.0596187941097273\\
0.41208904109589	0.0599958933194391\\
0.414109589041096	0.060372266094408\\
0.416130136986301	0.0607479149768715\\
0.418150684931507	0.0611228413096239\\
0.42	0.06146536096653\\
0.420171232876712	0.0614970451058003\\
0.422191780821918	0.0618705249221729\\
0.424212328767123	0.062243277736478\\
0.426232876712329	0.0626152988291953\\
0.428253424657534	0.0629865816702913\\
0.43	0.0633069197421064\\
0.43027397260274	0.0633571178114475\\
0.432294520547945	0.0637268967842132\\
0.434315068493151	0.0640959060046239\\
0.436335616438356	0.0644641306847693\\
0.438356164383562	0.0648315537518549\\
0.44	0.0651298694821039\\
0.440376712328767	0.0651981557751858\\
0.442397260273973	0.0655639149016956\\
0.444417808219178	0.0659288068004157\\
0.446438356164384	0.0662928046165588\\
0.448458904109589	0.0666558789355317\\
0.45	0.0669321595183646\\
0.450479452054795	0.0670179977575966\\
0.4525	0.0673791264835536\\
0.454520547945205	0.067739227912075\\
0.456541095890411	0.0680982622490969\\
0.458561643835616	0.06845618712991\\
0.46	0.0687102804815213\\
0.460582191780822	0.0688129576543526\\
0.462602739726027	0.0691685264357246\\
0.464623287671233	0.069522843663842\\
0.466643835616438	0.0698758571827974\\
0.468664383561644	0.0702275125839484\\
0.47	0.0704591897311797\\
0.470684931506849	0.0705777533145897\\
0.472705479452055	0.0709265208028159\\
0.47472602739726	0.0712737545991257\\
0.476746575342466	0.0716193925351848\\
0.478767123287671	0.0719633709002833\\
0.48	0.0721724130033919\\
0.480787671232877	0.0723056246359461\\
0.482808219178082	0.0726460875491519\\
0.484828767123288	0.0729846925446739\\
0.486849315068493	0.0733213718769751\\
0.488869863013699	0.0736560574220676\\
0.49	0.0738423595866997\\
0.490890410958904	0.0739886809698682\\
0.49291095890411	0.0743191745374061\\
0.494931506849315	0.0746474707033498\\
0.496952054794521	0.0749735029642021\\
0.498972602739726	0.0752972061127629\\
0.5	0.0754608872505163\\
0.500993150684932	0.0756185166389256\\
0.503013698630137	0.0759373731535425\\
0.505034246575342	0.0762537168354417\\
0.507054794520548	0.0765674919022293\\
0.509075342465753	0.076878646104937\\
0.51	0.0770201512504142\\
0.511095890410959	0.0771871312470907\\
0.513116438356164	0.0774929037284294\\
0.51513698630137	0.077795925113514\\
0.517157534246575	0.0780961627257305\\
0.519178082191781	0.0783935902666981\\
0.52	0.0785137690312837\\
0.521198630136986	0.0786881884615955\\
0.523219178082192	0.0789799457304932\\
0.525239726027397	0.0792688588860577\\
0.527260273972603	0.0795549338577199\\
0.529280821917808	0.0798381864426364\\
0.53	0.0799383287590644\\
0.531301369863014	0.0801186430835479\\
0.533321917808219	0.0803963416737815\\
0.535342465753425	0.0806713323894785\\
0.53736301369863	0.0809436785492536\\
0.539383561643836	0.0812134575014007\\
0.54	0.0812952655669607\\
0.541404109589041	0.0814807615386641\\
0.543424657534247	0.0817456988408341\\
0.545445205479452	0.0820083944450315\\
0.547465753424658	0.0822689912439297\\
0.549486301369863	0.0825276510117569\\
0.55	0.082593124443834\\
0.551506849315069	0.0827845554582745\\
0.553527397260274	0.0830399073105774\\
0.555547945205479	0.0832939314227535\\
0.557568493150685	0.0835468759133413\\
0.55958904109589	0.0837990133305992\\
0.56	0.0838502230764465\\
0.561609589041096	0.0840506418452973\\
0.563630136986301	0.0843020864711529\\
0.565650684931507	0.0845537003126579\\
0.567671232876712	0.0848058658401405\\
0.569691780821918	0.0850589961919574\\
0.57	0.0850977217474769\\
0.571712328767123	0.0853135365035543\\
0.573732876712329	0.0855699652632604\\
0.575753424657534	0.0858287956945086\\
0.57777397260274	0.0860905771642159\\
0.579794520547945	0.0863558966171874\\
0.58	0.0863831006628356\\
0.581815068493151	0.0866253800360399\\
0.583835616438356	0.0868996939264751\\
0.585856164383562	0.0871795468274841\\
0.587876712328767	0.0874656908461943\\
0.589897260273973	0.0877589232168406\\
0.59	0.0877740379200254\\
0.591917808219178	0.0880600878835868\\
0.593938356164384	0.0883700771066912\\
0.595958904109589	0.088689833091586\\
0.597979452054795	0.0890203496403703\\
0.6	0.089362673825259\\
0.6	0.089362673825259\\
};
\addlegendentry{\small box-LASSO, Analysis}

\addplot [color=blue, draw=none, mark size=3.5pt, mark=square, mark options={solid, blue}]
  table[row sep=crcr]{%
0.01	0.0294875\\
0.11	0.0124125\\
0.21	0.022925\\
0.31	0.03958125\\
0.41	0.0592625000000001\\
0.51	0.0760749999999999\\
};
\addlegendentry{\small box-LASSO, Simulations}

\end{axis}
\end{tikzpicture}%

%% file: Figs/LambdaSTARvsSNR.tex
%
%
\definecolor{mycolor1}{rgb}{0.00000,0.44700,0.74100}%
\begin{tikzpicture}

\begin{axis}[%
width=2.1in,
height=1.7in,
at={(1.989in,1.234in)},
scale only axis,
xmin=5,
xmax=15,
xlabel style={font=\color{white!15!black}},
xlabel={$\log\snr$ in [dB]},
ymin=.08,
ymax=.62,
ytick={0.128,0.3,.6},
yticklabels={{$0.13$},{$0.3$},{$0.6$}},
axis y line* = left,
ylabel style={font=\color{white!15!black}},
ylabel={$\lambda^\star$},
axis background/.style={fill=white}
]
\addplot [color=black, line width=1.0pt, forget plot]
  table[row sep=crcr]{%
5	0.547022325309085\\
5	0.547022325309085\\
5.03424657534247	0.542553456001139\\
5.06849315068493	0.538145885270197\\
5.1	0.534143958723477\\
5.1027397260274	0.53379833380548\\
5.13698630136986	0.529509558114219\\
5.17123287671233	0.525278349276058\\
5.2	0.521767760557567\\
5.20547945205479	0.521103531749043\\
5.23972602739726	0.516983962224724\\
5.27397260273973	0.512918528530015\\
5.3	0.509864338076191\\
5.30821917808219	0.508906148573625\\
5.34246575342466	0.504945769334936\\
5.37671232876712	0.501036365893363\\
5.4	0.498406578721287\\
5.41095890410959	0.497176940496322\\
5.44520547945205	0.493366521664007\\
5.47945205479452	0.489604163329314\\
5.5	0.487369433032513\\
5.51369863013699	0.485888944011296\\
5.54794520547945	0.482219966020642\\
5.58219178082192	0.478596354695746\\
5.6	0.476729721885942\\
5.61643835616438	0.475017257667984\\
5.65068493150685	0.47148184415493\\
5.68493150684932	0.467989304280257\\
5.7	0.466465964931415\\
5.71917808219178	0.464538848419165\\
5.75342465753425	0.46112970656823\\
5.78767123287671	0.457761127738609\\
5.8	0.456558227558242\\
5.82191780821918	0.454432379371601\\
5.85616438356164	0.451142746775607\\
5.89041095890411	0.447891532583595\\
5.9	0.44698798409657\\
5.92465753424658	0.444678056230175\\
5.95890410958904	0.441501653447497\\
5.99315068493151	0.438361675779156\\
6	0.437737995280314\\
6.02739726027397	0.43525749011138\\
6.06164383561644	0.43218847822078\\
6.0958904109589	0.42915403633798\\
6.1	0.428792198267234\\
6.13013698630137	0.426153574726485\\
6.16438356164384	0.423186517276179\\
6.1986301369863	0.420252301110837\\
6.2	0.420135607740576\\
6.23287671232877	0.417350376209129\\
6.26712328767123	0.41448020503854\\
6.3	0.411754226811655\\
6.3013698630137	0.411641262201725\\
6.33561643835616	0.40883303409479\\
6.36986301369863	0.406055018577046\\
6.4	0.403634966609084\\
6.4041095890411	0.403306724651781\\
6.43835616438356	0.400587672157626\\
6.47260273972603	0.397897391470107\\
6.5	0.395765573582882\\
6.50684931506849	0.395235423212998\\
6.54109589041096	0.392601317979094\\
6.57534246575342	0.389994636060051\\
6.6	0.388134563673941\\
6.60958904109589	0.387414947184956\\
6.64383561643836	0.38486183026729\\
6.67808219178082	0.382334873159984\\
6.7	0.380731162604597\\
6.71232876712329	0.379833672418254\\
6.74657534246575	0.377357833069942\\
6.78082191780822	0.374906968393079\\
6.8	0.373545251636859\\
6.81506849315068	0.372480699700405\\
6.84931506849315	0.370078656130606\\
6.88356164383562	0.367700474446009\\
6.9	0.366567318223358\\
6.91780821917808	0.365345798836516\\
6.95205479452055	0.363014280729547\\
6.98630136986301	0.360705578605779\\
7	0.359788411044153\\
7.02054794520548	0.358419357820483\\
7.05479452054795	0.356155290430251\\
7.08904109589041	0.353913055024932\\
7.1	0.353200098981669\\
7.12328767123288	0.351692336564593\\
7.15753424657534	0.34949282622133\\
7.19178082191781	0.347314221225766\\
7.2	0.3467944336375\\
7.22602739726027	0.345156224718068\\
7.26027397260274	0.343018545603333\\
7.29452054794521	0.340900898411199\\
7.3	0.340563915039765\\
7.32876712328767	0.338803003159522\\
7.36301369863014	0.336724585222005\\
7.3972602739726	0.33466537519963\\
7.4	0.33450146022893\\
7.43150684931507	0.332625108795763\\
7.46575342465753	0.330603526694833\\
7.5	0.32860037444444\\
7.5	0.32860037444444\\
7.53424657534247	0.326615402340793\\
7.56849315068493	0.324648365317372\\
7.6	0.322854324664675\\
7.6027397260274	0.322699022836698\\
7.63698630136986	0.320767138785119\\
7.67123287671233	0.318852481370515\\
7.7	0.317257315279313\\
7.70547945205479	0.316954823022814\\
7.73972602739726	0.315073940297258\\
7.77397260273973	0.313209613780293\\
7.8	0.311803665696531\\
7.80821917808219	0.311361627998038\\
7.84246575342466	0.309529771327226\\
7.87671232876712	0.307713835908545\\
7.9	0.306487989708155\\
7.91095890410959	0.305913617562311\\
7.94520547945205	0.304128915706396\\
7.97945205479452	0.302359533276338\\
8	0.301305176454077\\
8.01369863013699	0.30060527664757\\
8.04794520547945	0.298865955559701\\
8.08219178082192	0.297141383042786\\
8.1	0.296250372843436\\
8.11643835616438	0.295431375345529\\
8.15068493150685	0.293735751865348\\
8.18493150684932	0.292054335080265\\
8.2	0.291318967304358\\
8.21917808219178	0.290386950482549\\
8.25342465753425	0.288733426514062\\
8.28767123287671	0.287093594503272\\
8.3	0.286506574746808\\
8.32191780821918	0.285467288603863\\
8.35616438356164	0.283854345734909\\
8.39041095890411	0.28225460552256\\
8.4	0.281809022634399\\
8.42465753424658	0.280667910243197\\
8.45890410958904	0.279094104768015\\
8.49315068493151	0.277533036508987\\
8.5	0.277222338071118\\
8.52739726027397	0.275984555366174\\
8.56164383561644	0.274448513676345\\
8.5958904109589	0.272924766162859\\
8.6	0.272742735817897\\
8.63013698630137	0.271413169886782\\
8.66438356164384	0.269913584199201\\
8.6986301369863	0.2684258706947\\
8.7	0.268366607162043\\
8.73287671232877	0.26694989316597\\
8.76712328767123	0.265485517559509\\
8.8	0.264090509569691\\
8.8013698630137	0.264032611932399\\
8.83561643835616	0.262591046410111\\
8.86986301369863	0.261160693145332\\
8.9	0.25991115705795\\
8.9041095890411	0.259741426277756\\
8.93835616438356	0.258333121894846\\
8.97260273972603	0.256935657993508\\
9	0.255825411229132\\
9.00684931506849	0.255548914442679\\
9.04109589041096	0.254172772946789\\
9.07534246575342	0.252807117010072\\
9.1	0.251830272914719\\
9.10958904109589	0.251451831901714\\
9.14383561643836	0.250106804621805\\
9.17808219178082	0.248771923868083\\
9.2	0.247922874381349\\
9.21232876712329	0.247447080003435\\
9.24657534246575	0.246132165024151\\
9.28082191780822	0.244827072528898\\
9.3	0.24410047205535\\
9.31506849315068	0.243531697688402\\
9.34931506849315	0.242245937215821\\
9.38356164383562	0.240969689337779\\
9.4	0.240360439726118\\
9.41780821917808	0.239702853766062\\
9.45205479452055	0.238445331669936\\
9.48630136986301	0.237197025649102\\
9.5	0.236700262192093\\
9.52054794520548	0.235957839707231\\
9.55479452054795	0.234727679226107\\
9.58904109589041	0.233506450940325\\
9.6	0.233117529316165\\
9.62328767123288	0.232294062912557\\
9.65753424657534	0.231090424509356\\
9.69178082191781	0.229895446377491\\
9.7	0.229609930460157\\
9.72602739726027	0.228709040420795\\
9.76027397260274	0.22753111977752\\
9.7945205479452	0.226361598798178\\
9.8	0.226175249270574\\
9.82876712328767	0.225200393023862\\
9.86301369863014	0.224047419165032\\
9.8972602739726	0.222902595080756\\
9.9	0.222811358790096\\
9.93150684931507	0.221765839758393\\
9.96575342465754	0.220637073293708\\
10	0.21951621687141\\
10	0.21951621687141\\
10.0342465753425	0.218403192746099\\
10.0684931506849	0.217297924223614\\
10.1	0.216287861871845\\
10.1027397260274	0.216200335642777\\
10.1369863013699	0.215110352357507\\
10.1712328767123	0.214027900719323\\
10.2	0.213124408609037\\
10.2054794520548	0.212952908060196\\
10.2397260273973	0.211885302675762\\
10.2739726027397	0.210825013808883\\
10.3	0.210024044559391\\
10.3082191780822	0.209771971633541\\
10.3424657534247	0.208726107239064\\
10.3767123287671	0.207687352614677\\
10.4	0.206985026282573\\
10.4109589041096	0.206655640634364\\
10.4452054794521	0.205630905042036\\
10.4794520547945	0.204613080437002\\
10.5	0.204005676056547\\
10.513698630137	0.203602102259734\\
10.5479452054795	0.20259790677791\\
10.5821917808219	0.201600431072746\\
10.6	0.201084378708889\\
10.6164383561644	0.200609613025595\\
10.6506849315068	0.199625391304815\\
10.6849315068493	0.198647705352896\\
10.7	0.198219578631189\\
10.7191780821918	0.197676495373839\\
10.7534246575342	0.196711702320792\\
10.7876712328767	0.195753267883919\\
10.8	0.195409776964351\\
10.8219178082192	0.19480113447851\\
10.8561643835616	0.193855245233328\\
10.8904109589041	0.192915543979173\\
10.9	0.192653528943534\\
10.9246575342466	0.191981975237677\\
10.958904109589	0.191054484210308\\
10.9931506849315	0.190133016767592\\
11	0.189949441392288\\
11.027397260274	0.189217519438538\\
11.0616438356164	0.188307939400266\\
11.0958904109589	0.187404224467837\\
11.1	0.187296170356224\\
11.1301369863014	0.186506323084268\\
11.1643835616438	0.185614184310744\\
11.1986301369863	0.184727757817014\\
11.2	0.184692418867274\\
11.2328767123288	0.183846993871966\\
11.2671232876712	0.182971843334376\\
11.3	0.182136934830214\\
11.3013698630137	0.18210225764384\\
11.3356164383562	0.181238188811863\\
11.3698630136986	0.180379589413127\\
11.4	0.179628509023755\\
11.4041095890411	0.179526412576907\\
11.4383561643836	0.178678611978659\\
11.472602739726	0.177836141831757\\
11.5	0.177165973209021\\
11.5068493150685	0.176998956879376\\
11.541095890411	0.17616701238654\\
11.5753424657534	0.175340264132296\\
11.6	0.174748198338774\\
11.6095890410959	0.174518668402047\\
11.6438356164384	0.173702181980016\\
11.6780821917808	0.172890762141846\\
11.7	0.172374092861182\\
11.7123287671233	0.172084366647336\\
11.7465753424658	0.171282953733312\\
11.7808219178082	0.170486482106616\\
11.8	0.170042601112361\\
11.8150684931507	0.169694910937228\\
11.8493150684932	0.168908199851512\\
11.8835616438356	0.168126308925576\\
11.9	0.167752701792345\\
11.9178082191781	0.167349198678752\\
11.9520547945205	0.166576830067196\\
11.986301369863	0.165809164477594\\
12	0.165503406519454\\
12.0205479452055	0.165046163720984\\
12.0547945205479	0.164287790026684\\
12.0890410958904	0.163534006036326\\
12.1	0.163293758458423\\
12.1232876712329	0.162784774797994\\
12.1575342465753	0.162040059760467\\
12.1917808219178	0.161299824767555\\
12.2	0.161122831017923\\
12.2260273972603	0.160564034052542\\
12.2602739726027	0.159832652232718\\
12.2945205479452	0.159105644304003\\
12.3	0.158989726613425\\
12.3287671232877	0.158382975635675\\
12.3630136986301	0.157664611965173\\
12.3972602739726	0.156950519393001\\
12.4	0.156893575491596\\
12.4315068493151	0.156240664377709\\
12.4657534246575	0.155535013730965\\
12.5	0.154833534612711\\
12.5	0.154833534612711\\
12.5342465753425	0.154136194526393\\
12.5684931506849	0.153442961314282\\
12.6	0.152808786587733\\
12.6027397260274	0.152753803152863\\
12.6369863013699	0.152068688548308\\
12.6712328767123	0.151387586332021\\
12.7	0.150818538666988\\
12.7054794520548	0.150710465656261\\
12.7397260273973	0.150037295989829\\
12.7739726027397	0.149368047113838\\
12.8	0.148862021777507\\
12.8082191780822	0.148702689117541\\
12.8424657534247	0.148041192394236\\
12.8767123287671	0.147383527637236\\
12.9	0.146938489606344\\
12.9109589041096	0.146729665835899\\
12.9452054794521	0.146079578271734\\
12.9794520547945	0.145433236514557\\
13	0.145047217727291\\
13.013698630137	0.144790612418722\\
13.0479452054795	0.144151678119402\\
13.0821917808219	0.143516406028938\\
13.1	0.143187502768634\\
13.1164383561644	0.142884768833241\\
13.1506849315068	0.142256739488255\\
13.1849315068493	0.141632291216477\\
13.2	0.141358661619679\\
13.2191780821918	0.141011397503528\\
13.2534246575342	0.140394032094781\\
13.2876712328767	0.139780168992047\\
13.3	0.139560030673967\\
13.3219178082192	0.139169782450304\\
13.3561643835616	0.138562846974486\\
13.3904109589041	0.137959337316319\\
13.4	0.137790965107204\\
13.4246575342466	0.137359228471202\\
13.458904109589	0.136762495675149\\
13.4931506849315	0.136169114401762\\
13.5	0.136050838188033\\
13.527397260274	0.13557906035927\\
13.5616438356164	0.134992309487593\\
13.5958904109589	0.134408837955472\\
13.6	0.134339040619926\\
13.6301369863014	0.133828622157627\\
13.6643835616438	0.133251638711968\\
13.6986301369863	0.132677864456844\\
13.7	0.13265497991255\\
13.7328767123288	0.132107276448339\\
13.7671232876712	0.131539851957602\\
13.8	0.130998079781068\\
13.8013698630137	0.130975568468226\\
13.8356164383562	0.130414403673658\\
13.8698630136986	0.129856335474657\\
13.9	0.129367779571917\\
13.9041095890411	0.129301341976782\\
13.9383561643836	0.128749401487924\\
13.972602739726	0.128200492515874\\
14	0.127763533713715\\
14.0068493150685	0.127654593765922\\
14.041095890411	0.1271116841385\\
14.0753424657534	0.126571742726852\\
14.1	0.126184811191994\\
14.1095890410959	0.126034748814747\\
14.1438356164384	0.125500681874218\\
14.1780821917808	0.124969521563339\\
14.2	0.124631095046557\\
14.2123287671233	0.124441247724033\\
14.2465753424658	0.123915840379915\\
14.2808219178082	0.123393279734158\\
14.3	0.123101881890322\\
14.3150684931507	0.122873546167405\\
14.3493150684932	0.122356620235695\\
14.3835616438356	0.121842482668432\\
14.4	0.12159668144856\\
14.4178082191781	0.121331114366375\\
14.4520547945205	0.120822496399664\\
14.486301369863	0.120316610005865\\
14.5	0.120115016117533\\
14.5205479452055	0.119813436588056\\
14.5547945205479	0.119312957712928\\
14.5890410958904	0.118815155108919\\
14.6	0.11865642054155\\
14.6232876712329	0.118320010664377\\
14.6575342465753	0.117827506425743\\
14.6917808219178	0.117337624595764\\
14.7	0.117220441207563\\
14.7260273972603	0.116850347531732\\
14.7602739726027	0.11636565774374\\
14.7945205479452	0.115883537892978\\
14.8	0.115806636056421\\
14.8287671232877	0.115403970790033\\
14.8630136986301	0.114926939393233\\
14.8972602739726	0.114452426806995\\
14.9	0.114414574109995\\
14.9315068493151	0.113980416280214\\
14.9657534246575	0.11351089120466\\
15	0.113043835113402\\
15	0.113043835113402\\
};
\addplot [color=black, dashed, forget plot]
  table[row sep=crcr]{%
14	0\\
14	0.1278\\
};
\addplot [color=black, dashed, forget plot]
  table[row sep=crcr]{%
4	0.1278\\
14	0.1278\\
};
\addplot [color=black, draw=none, mark size=5pt, mark=x, mark options={solid, black}, forget plot]
  table[row sep=crcr]{%
14	0.1278\\
};
\end{axis}

\begin{axis}[%
width=2.1in,
height=1.7in,
at={(1.989in,1.234in)},
scale only axis,
xmin=5,
xmax=15,
ymin=-21,
ymax=-9.5,
hide x axis,
hide y axis,
legend style={legend cell align=left, align=left, draw=white!15!black}
]

\addplot [color=mycolor1, line width=1.0pt]
  table[row sep=crcr]{%
5	-10.7871832693258\\
5	-10.7871832693258\\
5.03424657534247	-10.8022011375867\\
5.06849315068493	-10.8173284056704\\
5.1	-10.8313418738292\\
5.1027397260274	-10.8325647941107\\
5.13698630136986	-10.8479100326529\\
5.17123287671233	-10.8633638600299\\
5.2	-10.8764288075006\\
5.20547945205479	-10.8789260237404\\
5.23972602739726	-10.8945962798286\\
5.27397260273973	-10.9103743926652\\
5.3	-10.9224377527289\\
5.30821917808219	-10.9262601347303\\
5.34246575342466	-10.9422532863973\\
5.37671232876712	-10.9583536357194\\
5.4	-10.9693629999201\\
5.41095890410959	-10.9745609782165\\
5.44520547945205	-10.9908751166645\\
5.47945205479452	-11.0072958608859\\
5.5	-11.0171994016427\\
5.51369863013699	-11.0238230275418\\
5.54794520547945	-11.0404564399259\\
5.58219178082192	-11.0571959277595\\
5.6	-11.0659423277476\\
5.61643835616438	-11.0740413269887\\
5.65068493150685	-11.0909924795827\\
5.68493150684932	-11.1080492333338\\
5.7	-11.1155876214638\\
5.71917808219178	-11.1252114416593\\
5.75342465753425	-11.1424789634041\\
5.78767123287671	-11.1598516626461\\
5.8	-11.1661315564684\\
5.82191780821918	-11.1773294085023\\
5.85616438356164	-11.1949120749364\\
5.89041095890411	-11.2125995405687\\
5.9	-11.2175707949331\\
5.92465753424658	-11.230391688487\\
5.95890410958904	-11.2482884060591\\
5.99315068493151	-11.2662895847472\\
6	-11.2699023465448\\
6.02739726027397	-11.2843951199236\\
6.06164383561644	-11.302604910688\\
6.0958904109589	-11.3209188596862\\
6.1	-11.3231235285021\\
6.13013698630137	-11.3393368729313\\
6.16438356164384	-11.357858859625\\
6.1986301369863	-11.3764847319814\\
6.2	-11.3772319264867\\
6.23287671232877	-11.3952144050527\\
6.26712328767123	-11.4140477965553\\
6.3	-11.4322253566101\\
6.3013698630137	-11.4329848266988\\
6.33561643835616	-11.4520254180152\\
6.36986301369863	-11.4711694951912\\
6.4	-11.4881018283357\\
6.4041095890411	-11.4904169849007\\
6.43835616438356	-11.5097678156397\\
6.47260273972603	-11.5292219175624\\
6.5	-11.5448595083758\\
6.50684931506849	-11.5487792223192\\
6.54109589041096	-11.5684396628958\\
6.57534246575342	-11.5882031734543\\
6.6	-11.6024966855639\\
6.60958904109589	-11.6080696891758\\
6.64383561643836	-11.6280391461041\\
6.67808219178082	-11.6481114809921\\
6.7	-11.6610117367027\\
6.71232876712329	-11.6682866311482\\
6.74657534246575	-11.6885645342856\\
6.78082191780822	-11.7089451283727\\
6.8	-11.7204030933866\\
6.81506849315068	-11.7294283514847\\
6.84931506849315	-11.7500141416576\\
6.88356164383562	-11.7707024367428\\
6.9	-11.7806692097999\\
6.91780821917808	-11.7914931742644\\
6.95205479452055	-11.8123862912767\\
6.98630136986301	-11.8333817242247\\
7	-11.8418085314902\\
7.02054794520548	-11.854479408805\\
7.05479452054795	-11.8756792798292\\
7.08904109589041	-11.896981271088\\
7.1	-11.9038194651165\\
7.12328767123288	-11.9183853152177\\
7.15753424657534	-11.9398913435676\\
7.19178082191781	-11.9614992860695\\
7.2	-11.9667003491735\\
7.22602739726027	-11.9832090711084\\
7.26027397260274	-12.0050206253947\\
7.29452054794521	-12.0269338738389\\
7.3	-12.0304494256901\\
7.32876712328767	-12.0489487394261\\
7.36301369863014	-12.0710651430942\\
7.3972602739726	-12.0932830036118\\
7.4	-12.0950648129038\\
7.43150684931507	-12.1156022374589\\
7.46575342465753	-12.1380227587085\\
7.5	-12.1605444789103\\
7.5	-12.1605444789103\\
7.53424657534247	-12.1831673069754\\
7.56849315068493	-12.2058911490632\\
7.6	-12.2268862162876\\
7.6027397260274	-12.2287159084691\\
7.63698630136986	-12.2516414855147\\
7.67123287671233	-12.2746677774387\\
7.7	-12.2940876176965\\
7.70547945205479	-12.29779467829\\
7.73972602739726	-12.321022078822\\
7.77397260273973	-12.3443498663887\\
7.8	-12.362146052455\\
7.80821917808219	-12.3677779248422\\
7.84246575342466	-12.391306134432\\
7.87671232876712	-12.4149343717055\\
7.9	-12.4310586440892\\
7.91095890410959	-12.4386625094105\\
7.94520547945205	-12.4624904163989\\
7.97945205479452	-12.4864179575324\\
8	-12.5008222488586\\
8.01369863013699	-12.5104449935894\\
8.04794520547945	-12.5345713811732\\
8.08219178082192	-12.558796972623\\
8.1	-12.5714334352568\\
8.11643835616438	-12.5831216159249\\
8.15068493150685	-12.6075451546255\\
8.18493150684932	-12.6320674277469\\
8.2	-12.6428884644876\\
8.21917808219178	-12.656688269703\\
8.25342465753425	-12.6814075102175\\
8.28767123287671	-12.7062249742435\\
8.3	-12.7151832719159\\
8.32191780821918	-12.7311404818851\\
8.35616438356164	-12.7561538483194\\
8.39041095890411	-12.7812648837215\\
8.4	-12.7883134494947\\
8.42465753424658	-12.8064733931901\\
8.45890410958904	-12.8317791766749\\
8.49315068493151	-12.8571820289059\\
8.5	-12.8622742291663\\
8.52739726027397	-12.8826817393238\\
8.56164383561644	-12.9082780920123\\
8.5958904109589	-12.9339708656319\\
8.6	-12.937060467239\\
8.63013698630137	-12.9597598333552\\
8.66438356164384	-12.9856447628035\\
8.6986301369863	-13.011625415986\\
8.7	-13.01266662974\\
8.73287671232877	-13.0377015492392\\
8.76712328767123	-13.0638729131689\\
8.8	-13.0890867787418\\
8.8013698630137	-13.0901392525931\\
8.83561643835616	-13.1165003064872\\
8.86986301369863	-13.1429558079299\\
8.9	-13.1663145596653\\
8.9041095890411	-13.1695054840516\\
8.93835616438356	-13.1961490559833\\
8.97260273972603	-13.2228862388084\\
9	-13.2443431895575\\
9.00684931506849	-13.2497167415147\\
9.04109589041096	-13.2766402669489\\
9.07534246575342	-13.3036565117725\\
9.1	-13.3231654463443\\
9.10958904109589	-13.330765166419\\
9.14383561643836	-13.3579659150528\\
9.17808219178082	-13.38525843553\\
9.2	-13.4027736590589\\
9.21232876712329	-13.41264239936\\
9.24657534246575	-13.4401174716697\\
9.28082191780822	-13.4676833111683\\
9.3	-13.4831596990452\\
9.31506849315068	-13.4953395701143\\
9.34931506849315	-13.5230858942839\\
9.38356164383562	-13.5509219229411\\
9.4	-13.5643149721366\\
9.41780821917808	-13.5788472888086\\
9.45205479452055	-13.6068616180418\\
9.48630136986301	-13.6349645302027\\
9.5	-13.6462304118094\\
9.52054794520548	-13.6631556382365\\
9.55479452054795	-13.6914345484496\\
9.58904109589041	-13.7198008604884\\
9.6	-13.7288964733125\\
9.62328767123288	-13.748254167321\\
9.65753424657534	-13.7767940552192\\
9.69178082191781	-13.8054201037427\\
9.7	-13.8123031287716\\
9.72602739726027	-13.8341318857249\\
9.76027397260274	-13.8629289672604\\
9.7945205479452	-13.8918109076931\\
9.8	-13.8964398632684\\
9.82876712328767	-13.9207772596075\\
9.86301369863014	-13.94982756882\\
9.8972602739726	-13.9789613743728\\
9.9	-13.9812956718959\\
9.93150684931507	-14.0081782085287\\
9.96575342465754	-14.0374775967684\\
10	-14.066859057788\\
10	-14.066859057788\\
10.0342465753425	-14.0963221034996\\
10.0684931506849	-14.1258662390321\\
10.1	-14.153118031125\\
10.1027397260274	-14.1554909627345\\
10.1369863013699	-14.1851957661804\\
10.1712328767123	-14.2149801341743\\
10.2	-14.240060109114\\
10.2054794520548	-14.244843544759\\
10.2397260273973	-14.274785469225\\
10.2739726027397	-14.3048053721216\\
10.3	-14.3276723169443\\
10.3082191780822	-14.334902711269\\
10.3424657534247	-14.3650769377724\\
10.3767123287671	-14.3953274960377\\
10.4	-14.4159411897189\\
10.4109589041096	-14.4256538237887\\
10.4452054794521	-14.4560553520854\\
10.4794520547945	-14.4865315053447\\
10.5	-14.5048527753601\\
10.513698630137	-14.5170817013623\\
10.5479452054795	-14.5477053513359\\
10.5821917808219	-14.5784018598903\\
10.6	-14.5943926384911\\
10.6164383561644	-14.6091706251041\\
10.6506849315068	-14.6400110385376\\
10.6849315068493	-14.6709224852629\\
10.7	-14.6845458652924\\
10.7191780821918	-14.7019043438949\\
10.7534246575342	-14.7329559866245\\
10.7876712328767	-14.7640767792527\\
10.8	-14.7752970693339\\
10.8219178082192	-14.7952660812268\\
10.8561643835616	-14.826523245678\\
10.8904109589041	-14.8578476194607\\
10.9	-14.8666303983817\\
10.9246575342466	-14.889238543193\\
10.958904109589	-14.9206953512993\\
10.9931506849315	-14.9522173720538\\
11	-14.9585295421804\\
11.027397260274	-14.9838039276266\\
11.0616438356164	-15.0154543341299\\
11.0958904109589	-15.0471679016674\\
11.1	-15.0509777412106\\
11.1301369863014	-15.0789439343842\\
11.1643835616438	-15.1107817305183\\
11.1986301369863	-15.1426805824544\\
11.2	-15.1439577964215\\
11.2328767123288	-15.1746397767784\\
11.2671232876712	-15.2066585943341\\
11.3	-15.2374520799392\\
11.3013698630137	-15.2387363102814\\
11.3356164383562	-15.2708721941556\\
11.3698630136986	-15.3030655099287\\
11.4	-15.3314425467491\\
11.4041095890411	-15.3353155160723\\
11.4383561643836	-15.367621465622\\
11.472602739726	-15.3999826062429\\
11.5	-15.4259107473551\\
11.5068493150685	-15.4323981802978\\
11.541095890411	-15.4648674249155\\
11.5753424657534	-15.497389572062\\
11.6	-15.5208378414128\\
11.6095890410959	-15.5299638486125\\
11.6438356164384	-15.5625894764252\\
11.6780821917808	-15.5952656724168\\
11.7	-15.6162046123381\\
11.7123287671233	-15.6279916486391\\
11.7465753424658	-15.6607666123578\\
11.7808219178082	-15.6935897661327\\
11.8	-15.7119914828918\\
11.8150684931507	-15.7264603078988\\
11.8493150684932	-15.7593774310503\\
11.8835616438356	-15.7923403245247\\
11.9	-15.8081785317385\\
11.9178082191781	-15.8253481728895\\
11.9520547945205	-15.8584001564303\\
11.986301369863	-15.89149545124\\
12	-15.904745510981\\
12.0205479452055	-15.9246332293101\\
12.0547945205479	-15.9578126586231\\
12.0890410958904	-15.9910329032473\\
12.1	-16.0016718646702\\
12.1232876712329	-16.024293123432\\
12.1575342465753	-16.0575924757053\\
12.1917808219178	-16.0909301129728\\
12.2	-16.0989367482903\\
12.2260273972603	-16.1243051846184\\
12.2602739726027	-16.157716836606\\
12.2945205479452	-16.1911642115836\\
12.3	-16.1965190492183\\
12.3287671232877	-16.2246464489886\\
12.3630136986301	-16.258162685154\\
12.3972602739726	-16.2917120534176\\
12.4	-16.29439740816\\
12.4315068493151	-16.3252936842313\\
12.4657534246575	-16.3589067052732\\
12.5	-16.3925502415604\\
12.5	-16.3925502415604\\
12.5342465753425	-16.4262234155636\\
12.5684931506849	-16.4599253473237\\
12.6	-16.4909557649894\\
12.6027397260274	-16.4936551545695\\
12.6369863013699	-16.527411952837\\
12.6712328767123	-16.5611948555904\\
12.7	-16.589592017503\\
12.7054794520548	-16.5950029743449\\
12.7397260273973	-16.6288354187907\\
12.7739726027397	-16.6626912969185\\
12.8	-16.6884368869797\\
12.8082191780822	-16.6965697151471\\
12.8424657534247	-16.730469778452\\
12.8767123287671	-16.764390590496\\
12.9	-16.7874681364317\\
12.9109589041096	-16.7983312537613\\
12.9452054794521	-16.8322908696826\\
12.9794520547945	-16.8662685387829\\
13	-16.886663431292\\
13.013698630137	-16.9002633608096\\
13.0479452054795	-16.9342744348733\\
13.0821917808219	-16.9683008595872\\
13.1	-16.9860003676761\\
13.1164383561644	-17.0023417332092\\
13.1506849315068	-17.036396153784\\
13.1849315068493	-17.0704632192886\\
13.2	-17.0854565016191\\
13.2191780821918	-17.1045420277778\\
13.2534246575342	-17.1386316775323\\
13.2876712328767	-17.1727312672079\\
13.3	-17.1850093792884\\
13.3219178082192	-17.2068398959863\\
13.3561643835616	-17.2409566637276\\
13.3904109589041	-17.2750806711246\\
13.4	-17.2846365681713\\
13.4246575342466	-17.3092110198577\\
13.458904109589	-17.343346812753\\
13.4931506849315	-17.3774871539404\\
13.5	-17.3843156892376\\
13.527397260274	-17.4116311490138\\
13.5616438356164	-17.4457779051938\\
13.5958904109589	-17.4799265314903\\
13.6	-17.4840244500777\\
13.6301369863014	-17.5140761388679\\
13.6643835616438	-17.5482258404128\\
13.6986301369863	-17.5823747515003\\
13.7	-17.5837406790162\\
13.7328767123288	-17.6165219899652\\
13.7671232876712	-17.6506666762728\\
13.8	-17.6834423602003\\
13.8013698630137	-17.684807933692\\
13.8356164383562	-17.7189448884695\\
13.8698630136986	-17.7530766700061\\
13.9	-17.7831076696633\\
13.9041095890411	-17.7872024110344\\
13.9383561643836	-17.8213212477975\\
13.972602739726	-17.8554323202305\\
14	-17.8827150123639\\
14.0068493150685	-17.8895347721423\\
14.041095890411	-17.9236277513996\\
14.0753424657534	-17.9577104101125\\
14.1	-17.9822430602003\\
14.1095890410959	-17.9917819048218\\
14.1438356164384	-18.0258413966871\\
14.1780821917808	-18.0598880516774\\
14.2	-18.0816707909995\\
14.2123287671233	-18.0939210407628\\
14.2465753424658	-18.1279395401078\\
14.2808219178082	-18.1619427312663\\
14.3	-18.1809775284821\\
14.3150684931507	-18.1959298013781\\
14.3493150684932	-18.2298999433669\\
14.3835616438356	-18.2638523561401\\
14.4	-18.280142983202\\
14.4178082191781	-18.2977862447899\\
14.4520547945205	-18.3317008207961\\
14.486301369863	-18.3655953022305\\
14.5	-18.3791472944618\\
14.5205479452055	-18.3994689139628\\
14.5547945205479	-18.4333208878685\\
14.5890410958904	-18.4671504630372\\
14.6	-18.4779710732027\\
14.6232876712329	-18.5009568859841\\
14.6575342465753	-18.5347394108616\\
14.6917808219178	-18.5684972996735\\
14.7	-18.5765954458704\\
14.7260273972603	-18.6022298224903\\
14.7602739726027	-18.6359362576661\\
14.7945205479452	-18.669615892057\\
14.8	-18.6750020992558\\
14.8287671232877	-18.7032680212417\\
14.8630136986301	-18.7368919497425\\
14.8972602739726	-18.770486991249\\
14.9	-18.7731733263107\\
14.9315068493151	-18.8040524688429\\
14.9657534246575	-18.8375877152242\\
15	-18.8710920729394\\
15	-18.8710920729394\\
};
\addlegendentry{\small Analysis}

\addplot [color=mycolor1, draw=none, mark size=2.5pt, mark=square, mark options={solid, mycolor1}]
  table[row sep=crcr]{%
6	-11.2577462363144\\
8	-12.39875404647\\
10	-13.9835705381357\\
12	-15.7823709188756\\
14	-17.7710754986284\\
};
\addlegendentry{\small Simulations}

\addplot [color=black, draw=none, mark size=5.0pt, mark=x, mark options={solid, black}, forget plot]
  table[row sep=crcr]{%
14	-17.88\\
};
\addplot [color=black, dashed, forget plot]
  table[row sep=crcr]{%
14	-20\\
14	-17.88\\
};
\addplot [color=black, dashed, forget plot]
  table[row sep=crcr]{%
16	-17.88\\
14	-17.88\\
};

\end{axis}
\pgfplotsset{every axis y label/.append style={rotate=180,yshift=8.1cm}}
\begin{axis}[%
width=2.1in,
height=1.7in,
at={(1.989in,1.234in)},
scale only axis,
xmin=5,
xmax=15,
xlabel style={font=\color{white!15!black}},
xlabel={$\log\snr$ in [dB]},
ymin=-21,
ymax=-9.5,
ytick={-20,-17.88,-15,-10},
yticklabels={{\textcolor{mycolor1}{$-16$}},{\textcolor{mycolor1}{$-17.9$}},{\textcolor{mycolor1}{$-15$}},{\textcolor{mycolor1}{$-10$}}},
hide x axis,
axis y line*={right, draw=mycolor1},
ylabel style={font=\color{mycolor1}},
ylabel={$\log \mse$ in [dB]},
]
\end{axis}

\end{tikzpicture}%